\documentclass[smallextended]{svjour3}       
\smartqed  
\usepackage{graphicx}

\usepackage{mathptmx}      
\usepackage{graphics} 
\usepackage{epsfig} 
\usepackage{mathptmx} 
\usepackage{times} 
\usepackage{amsmath} 
\usepackage{amssymb}  

\usepackage{amsthm}
\usepackage{comment}
\usepackage{color}
\usepackage{xcolor}
\usepackage{graphicx}
\usepackage{float}
\usepackage[]{algorithmic}
\usepackage[justification=centering]{caption}

\newtheorem{algorithm}{Procedure}
\newtheorem{assumption}{Assumption}

\floatname{algorithm}{Procedure}

\journalname{Discrete Event Dynamic Systems}

\newcommand{\MR}[1]{\textcolor{blue}{MR: #1}}
\begin{document}

\title{Synthesis of Successful Actuator Attackers on Supervisors
%
}


\author{
Liyong Lin \and Sander Thuijsman \and Yuting Zhu \and Simon Ware \and Rong Su \and Michel Reniers
}

\authorrunning{L. Lin, S.B. Thuijsman, Y. Zhu, S. Ware, R. Su, M.A. Reniers} 

\institute{Liyong Lin, Yuting Zhu, Simon Ware, Rong Su  \at
              Electrical and Electronic Engineering, Nanyang Technological University, Singapore \\
\email{liyong.lin@ntu.edu.sg, yuting002@e.ntu.edu.sg, sware@ntu.edu.sg, rsu@ntu.edu.sg} 
           \and
            Sander Thuijsman, Michel Reniers \at
              Mechanical Engineering Department,  Eindhoven University of Technology, Netherlands\\
              \email{s.b.thuijsman@student.tue.nl, m.a.reniers@tue.nl} \\\\
Liyong Lin and Sander Thuijsman contribute equally to this work.}

\date{Received: date / Accepted: date}

\maketitle

\begin{abstract}
In this work, we propose and develop a new discrete-event based actuator attack model on the closed-loop system formed by the plant and the supervisor. We assume the actuator attacker partially observes the execution of the closed-loop system and eavesdrops the control commands issued by the supervisor. The attacker can modify each control command on a specified subset of attackable events. The attack principle of the actuator attacker is to remain covert until it can establish a successful attack and lead the attacked closed-loop system into generating certain damaging strings. We present a characterization for the existence of a successful attacker, via a new notion of attackability, and prove the existence of the supremal successful actuator attacker, when both the supervisor and the attacker are normal (that is, unobservable events to the supervisor cannot be disabled by the supervisor and unobservable events to the attacker cannot be attacked by the attacker). Finally, we present an algorithm to synthesize the supremal successful attackers that are represented by Moore automata. 
\end{abstract}


\section{Introduction}

Recently, cyber-physical systems have drawn much research interest within the discrete-event systems and formal methods community~\cite{CarvalhoEnablementAttacks},~\cite{Su2018},~\cite{Goes2017},~\cite{Lanotte2017},~\cite{Jones2014},~\cite{WTH17},~\cite{LACM17}. Due to the prevalence of cyber-physical systems in modern society and the catastrophic consequences that could occur after these systems get attacked~\cite{MKBDLP12},~\cite{D13}, it is important to design resilient control mechanisms against adversarial attacks; the first step towards this goal is to understand when adversarial attacks can succeed. 


In this work, we focus on discrete-event systems as our model of cyber-physical systems and consider those adversaries that use actuator attacks. Most of the existing works in the discrete-event framework consider the problem of deception attack instead~\cite{Su2018},~\cite{Goes2017},~\cite{WTH17}, where the attacker would accomplish the attack goal by altering the sensor readings. The work that is most closely related to ours is~\cite{CarvalhoEnablementAttacks}, in which the authors present an actuator enablement attack scenario. For actuator enablement attack (AE attack), the attacker's goal is to drive the plant into an unsafe state by overwriting the supervisor's disablement actions on some attackable events with enablement actions. Our work is mainly inspired by~\cite{CarvalhoEnablementAttacks}; here we present some interesting  generalizations and enhancements over the setup presented in~\cite{CarvalhoEnablementAttacks}, which allows us to model a more powerful actuator attacker and consider a richer set of defense strategies. The main contributions of this work are as follows:
\begin{enumerate}
\item We present a formal formulation of the supervisor (augmented with a monitoring mechanism), attacker, attacked closed-loop system and the attack goal. Compared with~\cite{CarvalhoEnablementAttacks}, the attacker is not restricted to  having the same observation scope as the supervisor. Furthermore, the attacker  eavesdrops the control commands issued by the supervisor, which could be used to refine the knowledge of the attacker about the string generated by the plant. The actuator attack mechanism does not have to be AE attacks. The attacker may need to disable some enabled attackable events to fulfill the attack goal; indeed, in many cases, this is the only way to accomplish the attack goal due to partial observation by the attacker.
\item We then present a characterization for the existence of a successful attacker against a normal supervisor, under a normality assumption that all the attackable events are observable to the attacker; a notion of attackability plays an important role in the characterization. After that, we prove the existence of the supremal successful attacker and then present a characterization of it. 
\item We also provide an algorithm for the synthesis of the supremal successful attackers that are represented by Moore automata.
\end{enumerate}

The paper is organized as follows: Section~\ref{sec: prel} is devoted to the preliminaries. Then, in Section~\ref{section:SS}, we shall provide a detailed explanation about the system setup, including the supervisor (augmented with a monitoring mechanism), attacker, attacked closed-loop system and the attack goal. In Section~\ref{section: ASNS}, we present the characterization results, together with a synthesis algorithm for synthesizing the supremal successful attackers on normal supervisors, under the assumption that attackable events are observable to the attackers. Finally, discussions and conclusions are presented in Section~\ref{section:DC}.
\section{Preliminaries}
\label{sec: prel}
We assume the reader is familiar with the basics of supervisory control theory~\cite{WMW10},~\cite{CL99} and automata theory~\cite{HU79}. In this section, we shall recall the basic notions and terminology that are necessary to understand this work.

A finite state automaton is a 4-tuple $G=(Q, \Sigma, \delta, q_0)$, where $Q$ is the finite set of states, $\Sigma$ the alphabet of events, $\delta: Q \times \Sigma \rightarrow Q$ the partial transition function\footnote{As usual, $\delta$ is naturally extended to the partial transition function $\delta: Q \times \Sigma^* \rightarrow Q$. We write $\delta(q, s)!$ to mean that string $s$ is defined at state $q$. For any $Q' \subseteq Q, L \subseteq \Sigma^*$, we let $\delta(Q', L)=\{\delta(q, s) \mid q \in Q', s \in L, \delta(q, s)!\}$. We also view $\delta \subseteq Q \times \Sigma^* \times Q$ as a transition relation.}, $q_0 \in Q$ the initial state. A control constraint over $\Sigma$ is a tuple $(\Sigma_c, \Sigma_o)$ of sub-alphabets of $\Sigma$, where $\Sigma_o \subseteq \Sigma$ denotes the set of observable events and $\Sigma_c \subseteq \Sigma$ the set of controllable events. Let $\Sigma_{uo}:=\Sigma-\Sigma_o$ denote the set of unobservable events and $\Sigma_{uc}:=\Sigma-\Sigma_c$ the set of uncontrollable events. For each sub-alphabet $\Sigma' \subseteq \Sigma$, the natural projection $P: \Sigma^* \rightarrow \Sigma'^*$ is defined and naturally extended to a mapping between languages~\cite{WMW10}. Let $L(G)=\{s \in \Sigma^* \mid \delta(q_0, s)!\}$ denote the closed-behavior of $G$. A supervisor on $G$ w.r.t. control constraint $(\Sigma_c, \Sigma_o)$ is a map $V: P_{o}(L(G)) \rightarrow \Gamma$, where $P_o: \Sigma^* \rightarrow \Sigma_o^*$ is the natural projection and $\Gamma:=\{\gamma \subseteq \Sigma \mid \Sigma_{uc} \subseteq \gamma\}$. $V(w)$ is the control command issued by the supervisor after observing $w \in P_{o}(L(G))$. $V/G$ denotes the closed-loop system of $G$ under the supervision of $V$. The closed-behavior $L(V/G)$ of $V/G$ is inductively defined as follows:
\begin{enumerate}
\item $\epsilon \in L(V/G)$,
\item if $s \in L(V/G)$,  $\sigma \in V(P_{o}(s))$ and $s\sigma \in L(G)$, then $s\sigma \in L(V/G)$,
\item no other strings belong to $L(V/G)$.
\end{enumerate}

For any two strings $s, t \in \Sigma^*$, we write $s \leq t$ (respectively, $s<t$) to denote that $s$ is a prefix (respectively, strict prefix) of $t$. For any language $L \subseteq \Sigma^*$, $\overline{L}$ is used to denote the prefix-closure of $L$. For two finite state automata $G_1, G_2$, we write $G_1 \lVert G_2$ to denote their synchronous product. For any string $s$, we write $|s|$ to denote the length of $s$. For convenience, we often identify a singleton with the unique element it contains.

\section{System Setup}
\label{section:SS}
In this section, we shall introduce the system setup that will be followed in this work. Intuitive explanations about the attack architecture and assumptions are provided in Section~\ref{subs:BI}. After that, a formal setup is provided in Section~\ref{subs: FS}, which includes a formalization of the supervisor (augmented with a monitoring mechanism), attacker, attacked closed-loop system and the attack goal.
\subsection{Basic Ideas}
\label{subs:BI}
Let us consider the architecture for actuator attack that is shown in Fig.~\ref{fig:ring}. The plant $G$, which is a finite state automaton over $\Sigma$, is under the control of a (partially observing) supervisor $V$ w.r.t. $(\Sigma_c, \Sigma_o)$, which is also given by a finite state automaton over $\Sigma$. In addition, we assume the existence of an actuator attacker $A$ that is able to modify the control command $\gamma$ issued by the supervisor on a designated set of attackable events $\Sigma_{c, A} \subseteq \Sigma_c$, each time when it intercepts a control command\footnote{However, the attacker is not allowed to delete existing control commands or insert new control commands. The attacker cannot attack those events outside $\Sigma_{c, A}$; thus, by allowing the attacker to delete existing control commands or insert new control commands, the attacker will not be made more powerful.}. The plant $G$ follows the modified control command $\gamma'$ instead of $\gamma$. We assume the supervisor is augmented with a monitoring mechanism that monitors the execution of the closed-loop system (under attack)\footnote{As we shall see in Section~\ref{subsubs: S}, the monitoring mechanism is embedded into the supervisor $V$. Thus, $V$ is not only responsible for controlling the plant but also halting the execution of the closed-loop system after detecting an actuator attack. For more details, see our formalization of the system setup provided in Section~\ref{subs: FS}.}. If $\sigma \in \Sigma_{c, A}$ is enabled by the supervisor in $\gamma$, then the attacker is not discovered\footnote{The supervisor is not sure whether $\sigma$ has been disabled by an attacker, even if disabling $\sigma$ may result in deadlock, as the supervisor is never sure whether: 1) deadlock has occurred due to actuator attack, or 2) $\sigma$ will possibly fire soon (according to the internal mechanism of the plant).} by the supervisor when it disables $\sigma$ in $\gamma'$; on the other hand, if $\sigma \in \Sigma_{c, A}$ is disabled in $\gamma$, then the decision of enabling $\sigma$ in $\gamma'$ would place the attacker in the risk of being discovered immediately by the supervisor if $\sigma$ can be fired\footnote{Once $\sigma$ is fired, the supervisor immediately knows that an actuator attack has occurred if $\sigma$ is observable to the supervisor, since $\sigma$ is supposed to have been disabled. $\sigma$ can be fired only if it is enabled by the plant.}, unless $\sigma$ is unobservable to the supervisor. 

\begin{figure}[h]
\centering
\hspace*{-1mm}
\captionsetup{justification=centering}
\includegraphics[width=2.6in, height=1.4in]{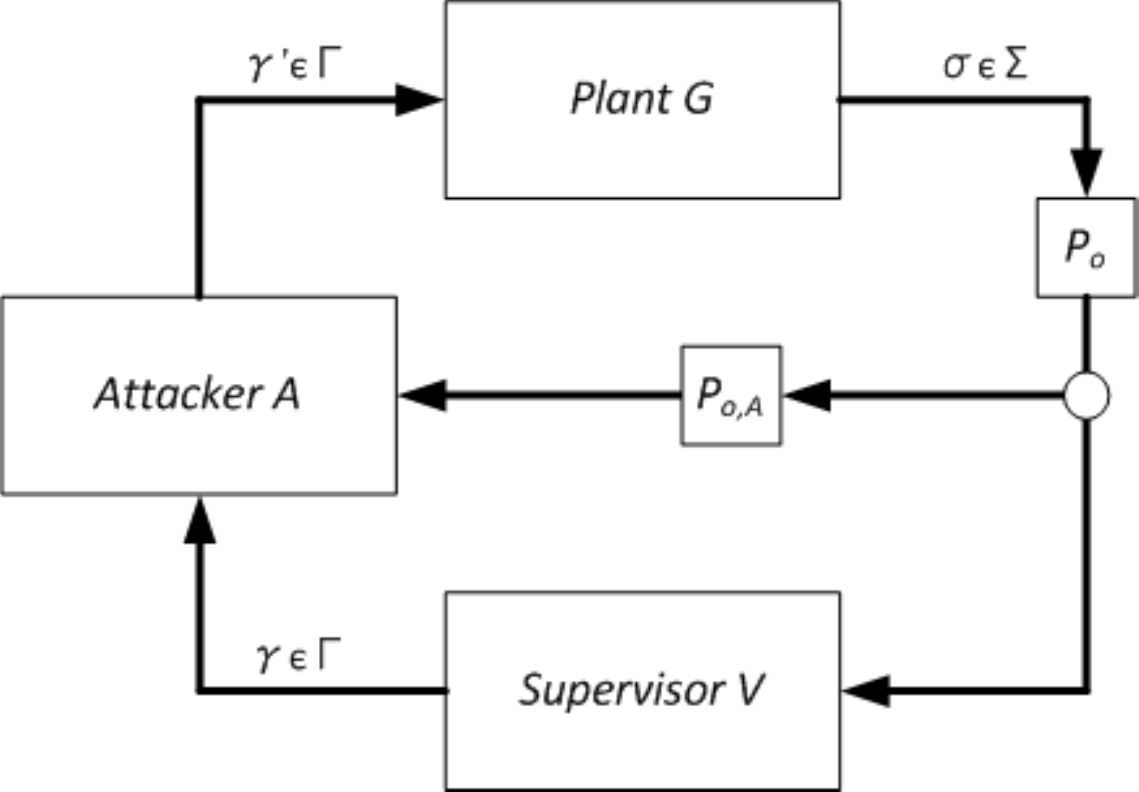} 
\caption{The architecture for actuator attack}
\label{fig:ring}
\end{figure}

We assume that the attacker has a complete knowledge about the models of the plant $G$ and the supervisor $V$ (and its realization $S$ using a finite state automaton), including the control constraint $(\Sigma_{c}, \Sigma_{o})$ over $\Sigma$. The attacker $A$ partially observes the execution of the closed-loop system and (at the same time) eavesdrops the sequence of control commands issued by the supervisor. We assume that the set $\Sigma_{o, A}$ of events observable to the attacker is a subset of the set $\Sigma_o$ of events observable to the supervisor. In other words, we assume that the attacker cannot deploy additional sensors than those that are available to the supervisor. Thus, the supervisor may have a better knowledge about the execution of the closed-loop system than the attacker does. The sequence of control commands (issued by the supervisor) encodes to a certain extent the supervisor's knowledge on the execution of the closed-loop system, which could be used by the attacker to better estimate the string generated by the plant and determine whether it shall establish an attack on some attackable events. We assume that each time when the supervisor observes an event (and makes a state transition), it issues a new control command that can be intercepted by the attacker. By observing a control command from the supervisor, the attacker concludes that an observable event $\sigma$ of the supervisor has occurred even if $\sigma$ may be unobservable to the attacker. 
 
The goal of the attacker is to drive the attacked closed-loop system into executing certain  damaging strings outside the controlled behavior. The supervisor has a mechanism for halting the execution of the closed-loop system after discovering an actuator attack. Thus, in order to attack the closed-loop system successfully, the attacker must remain covert\footnote{The attacker is covert if the supervisor is not sure whether an attack has occurred.}~\cite{Su2018} and enable some disabled (attackable) events only when either 1) the attack will not be discovered by the supervisor, i.e., the attacker remains covert, or 2) if the attack is discovered by the supervisor, then a damaging string has been generated. In particular, if $\sigma \in \Sigma_{c, A}$ is disabled by the supervisor in $\gamma$, then enabling $\sigma$ in $\gamma'$ will not place the attacker in the risk of being discovered by the supervisor if 1) $\sigma$ is unobservable to the supervisor, or 2) if $\sigma$ is not enabled by the plant and thus cannot be fired.\\


\subsection{Formal Setup}
\label{subs: FS}
In this section, we explain the formal setup that is used in this work. In Section~\ref{subsubs: S}, we provide a formal definition of a supervisor, which is augmented with a monitoring mechanism that will halt the execution of the (attacked) closed-loop system once an actuator attack is discovered. In Section~\ref{subsubs: AA}, we consider an actuator attacker that is able to partially observe the execution of the closed-loop system and (at the same time) eavesdrop the control commands issued by the supervisor. In Section~\ref{subsubs: ACLS}, we then introduce the notion of a damage-inflicting set and explain the attacked closed-loop system formed by the plant, supervisor and actuator attacker. Then, we present a formulation of the goal of the actuator attacker. 
\subsubsection{Supervisor}
\label{subsubs: S}
We recall that a supervisor on $G$ w.r.t. $(\Sigma_c, \Sigma_o)$ is a map $V: P_{o}(L(G)) \rightarrow \Gamma$, where $\Gamma=\{\gamma \subseteq \Sigma \mid \Sigma_{uc} \subseteq \gamma\}$. Since $P_{o}(L(G))=(P_{o}(L(G))-P_{o}(L(V/G))) \dot{\cup} P_{o}(L(V/G))$ and only those strings in $P_{o}(L(V/G))$ can be observed by the supervisor in the normal execution of the closed-loop system, the standard definition of a supervisor requires that the supervisor shall also apply control for those strings $w \in P_{o}(L(G))-P_{o}(L(V/G))$ that cannot be observed in the normal execution of the closed-loop system. From the point of view of the supervisor, it will conclude the existence of an (actuator) attacker and then halt the execution of the closed-loop system the first time when it observes some string $w \notin P_{o}(L(V/G))$. Under this interpretation, the supervisor is augmented with a monitoring mechanism for the detection of actuator attack and does not control outside $P_{o}(L(V/G))$; thus, the supervisor (augmented with a monitoring mechanism) is effectively\footnote{To explicitly model the monitoring mechanism, we could pad the domain of $V$ to $P_{o}(L(G))$ by setting $V(w)=$`halt' for any $w \in P_{o}(L(G))-P_{o}(L(V/G))$, where `halt' is a special control symbol that immediately halts the execution of the closed-loop system. This view will be implicitly adopted in the rest.} a map $V: P_{o}(L(V/G)) \rightarrow \Gamma$. 

\subsubsection{Actuator Attacker}
\label{subsubs: AA}
In this section, we shall present a formal definition of an actuator attacker. To that end, we need to model the partial observation capability of an attacker, taking into account of the fact that control commands issued by the supervisor can be eavesdropped. 

We here note that an actuator attacker may be able to attack the closed-loop system only if the supervisor has not halted the execution (of the closed-loop system); suppose $s$ is the string that is generated by the plant, then the attacker may be able to attack the closed-loop system only if $P_o(s) \in P_o(L(V/G))$ (see Section~\ref{subsubs: S}). It is important to understand the observation sequence for the attacker when $s$ is generated, which is the only factor that determines the attack decision of the attacker. If the plant executes any event $\sigma \in \Sigma_{uo}$ that is unobservable to the supervisor, then the attacker will observe nothing since $\sigma$ is unobservable to the attacker (no observation of execution from the plant) and $\sigma$ is unobservable to the supervisor (no control command will be issued by the supervisor). Thus, the attacker's observation sequence solely depends on the supervisor's observation sequence $P_{o}(s) \in P_o(L(V/G))$. Now, we still need to define a function $\hat{P}_{o, A}^{V}$ that maps supervisor's observation sequences to attacker's observation sequences. Let $P_o(s)=\sigma_1\sigma_2\ldots \sigma_n$ denote the supervisor's observation sequence. We shall explain how it can be transformed into the attacker's observation sequence. If the supervisor observes the execution of event $\sigma_i \in \Sigma_{o}$, then the supervisor will issue a new control command $V(\sigma_1\sigma_2\ldots \sigma_i)$, since $\sigma_i$ is observable to the supervisor, which is then intercepted by the attacker. There are two subcases for $\sigma_i \in \Sigma_o$, regarding the attacker's observation on the execution of the plant. If $\sigma_i \in \Sigma_{o, A} \subseteq \Sigma_o$, then the attacker will also observe execution of $\sigma_i$ from the plant. If $\sigma_i \in \Sigma_o-\Sigma_{o, A}$, then the attacker will observe nothing (that is, $\epsilon$) from the plant; however, based on the newly issued control command, it can infer that some event in $\Sigma_o-\Sigma_{o, A}$ has already been fired. Thus, at each step, the attacker observes $(P_{o, A}(\sigma_i), V(\sigma_1\sigma_2\ldots \sigma_i))$. Then, the attacker's observation sequence is simply a string in $((\Sigma_{o, A} \cup \{\epsilon\}) \times \Gamma)^*$. From the above discussion, we obtain that $\hat{P}_{o, A}^{V}: P_o(L(V/G)) \rightarrow ((\Sigma_{o, A} \cup \{\epsilon\}) \times \Gamma)^*$ is a function that maps supervisor's observation sequences to attacker's observation sequences, where, for any $w=\sigma_1\sigma_2\ldots \sigma_n \in P_o(L(V/G))$, $\hat{P}_{o, A}^{V}(w)$ is defined to be
\begin{center}
 $(P_{o, A}(\sigma_1), V(\sigma_1))(P_{o, A}(\sigma_2), V(\sigma_1\sigma_2))\ldots (P_{o, A}(\sigma_n), V(\sigma_1\sigma_2 \ldots \sigma_n)) \in ((\Sigma_{o, A} \cup \{\epsilon\}) \times \Gamma)^*$.
\end{center}
The set of all the possible observation sequences\footnote{The execution of the attacked closed-loop system will be halted once $w \notin P_o(L(V/G))$ is observed by the supervisor, at which point the attacker will not (need to) record any further observation sequence.} for the attacker is $\hat{P}_{o, A}^{V}(P_o(L(V/G)))$. The attacker modifies the control command it intercepts based upon its observation. We shall call the tuple $(\Sigma_{c, A}, \Sigma_{o, A})$ an {\em attack constraint} over $\Sigma$. An attacker on $(G, V)$ w.r.t. attack constraint $(\Sigma_{c, A}, \Sigma_{o, A})$ is a function $A: \hat{P}_{o, A}^{V}(P_o(L(V/G))) \rightarrow \Delta \Gamma$, where $\Delta \Gamma=2^{\Sigma_{c, A}}$. Here, $\Delta \Gamma$ denotes the set of all the possible attack decisions that can be made by the attacker on the set $\Sigma_{c, A}$ of attackable events. Intuitively, each $\Delta \gamma \in \Delta \Gamma$ denotes the set of enabled attackable events that are determined by the attacker. For any $P_o(s) \in P_o(L(V/G))$, the attacker determines the set $\Delta \gamma=A(\hat{P}_{o, A}^{V}(P_o(s)))$ of attackable events to be enabled based on its observation $\hat{P}_{o, A}^{V}(P_o(s))$. We notice that the set of attackable events enabled by the supervisor in $\gamma$ is $V(P_o(s)) \cap \Sigma_{c, A}$. Thus, the attacker applies {\em actual attack} iff $V(P_o(s))\cap \Sigma_{c, A} \neq A(\hat{P}_{o, A}^{V}(P_o(s)))$.

\subsubsection{Attacked Closed-loop System and Damage-Inflicting Set}
\label{subsubs: ACLS}
Based on the above formal definitions of a supervisor and an actuator attacker, we are ready to define the attacked closed-loop system formed by the plant, supervisor and attacker. 

We already know that, after string $s$ is generated in the plant with $P_o(s) \in P_o(L(V/G))$, the control command issued by the supervisor is $\gamma=V(P_o(s))$ and the attack decision made by the actuator attacker is $\Delta \gamma=A(\hat{P}_{o, A}^{V}(P_o(s)))$. Then, it follows that the modified control command is $\gamma':=(\gamma -\Sigma_{c, A}) \cup \Delta \gamma=(V(P_o(s))-\Sigma_{c, A}) \cup A(\hat{P}_{o, A}^{V}(P_o(s)))$. Now, given the supervisor $V$ and the actuator attacker $A$, we could lump them together into an equivalent (attacked) supervisor $V_A: P_o(L(V/G)) \rightarrow \Gamma$ such that, for any string $w \in P_o(L(V/G))$ observed by the attacked supervisor, the control command issued by $V_A$ is $V_A(w)=(V(w)-\Sigma_{c, A}) \cup A(\hat{P}_{o, A}^{V}(w))$. The plant $G$ follows the control command issued by the attacked supervisor $V_A$ and the attacked closed-loop system is denoted by $V_A/G$ (the grouping of supervisor $V$ and attacker $A$ into $V_A$ is illustrated in Fig.~\ref{fig:attack}). In the presence of an actuator attacker performing actuator attacks, the execution of the closed-loop system will be immediately halted once the supervisor detects an attack from the actuator attacker. Thus, the definition of the (attacked) closed-behavior $L(V_A/G)$ of $V_A/G$ is slightly different from that of $L(V/G)$ (due to the existence of the monitoring mechanism) and is inductively defined as follows:
\begin{enumerate}
\item $\epsilon \in L(V_A/G)$,
\item if $s \in L(V_A/G)$, $P_o(s) \in P_o(L(V/G))$, $\sigma \in V_A(P_{o}(s))$ and $s\sigma \in L(G)$,  then $s\sigma \in L(V_A/G)$,
\item no other strings belong to $L(V_A/G)$.
\end{enumerate}
We notice that $s\sigma \in L(V_A/G))$ only if $P_o(s) \in P_o(L(V/G))$. That is, if $P_o(s) \notin P_o(L(V/G))$, the supervisor already halts the execution of the closed-loop system and thus $s\sigma$ cannot be generated. The goal of the actuator attacker is to drive the attacked closed-loop system into executing certain damaging strings. Let $L_{dmg} \subseteq L(G)$ be some regular language over $\Sigma$ that denotes a so-called ``damage-inflicting set", where each string is a damaging string. We here shall require that $L_{dmg} \cap L(V/G)=\varnothing$, that is, no string in $L(V/G)$ could be damaging. Then, we have that $L_{dmg} \subseteq L(G)-L(V/G)$. The goal of the actuator attacker is formulated as follows:
\begin{enumerate}
\item $L(V_A/G) \cap L_{dmg}\neq \varnothing$, and
\item $L(V_A/G)-P_{o}^{-1}P_o(L(V/G))\subseteq L_{dmg}\Sigma^*$.
\end{enumerate}
The first condition says that the actuator attacker can drive the system into executing certain damaging strings. It is a ``possibilistic" statement in the sense that it is possible for a damaging string to be generated by the attacked closed-loop system. However, there is no guarantee that a damaging string will be generated in any single run of the attacked closed-loop system, due to potential uncontrollable factors that are beyond the capability of the attacker. The second condition says that once the attacked closed-loop system runs outside $P_{o}^{-1}P_o(L(V/G))$, i.e., $P_o(s) \notin P_o(L(V/G)$, at which point the supervisor detects the actuator attack and halts the execution of the attacked closed-loop system, a damaging string must have been generated\footnote{That is, the attacker is only detected when it is too late.}, i.e. some prefix of $s$ is a damaging string. If the actuator attacker achieves its goal, then it is considered to be successful. Thus, we say an attacker $A$ is {\em successful} on $(G, V)$ w.r.t. $(\Sigma_{c, A}, \Sigma_{o, A})$ and $L_{dmg}$ if 
\begin{enumerate}
\item $L(V_A/G) \cap L_{dmg}\neq \varnothing$, and
\item $L(V_A/G)-P_{o}^{-1}P_o(L(V/G))\subseteq L_{dmg}\Sigma^*$.
\end{enumerate}
A supervisor $V$ is said to be {\em resilient} (against actuator attack) on $G$ w.r.t. $(\Sigma_{c, A}, \Sigma_{o, A})$ and $L_{dmg}$ if there is no successful actuator attacker on $(G, V)$ w.r.t. $(\Sigma_{c, A}, \Sigma_{o, A})$ and $L_{dmg}$. It is of interest to consider the following two basic problems:
\begin{enumerate}
\item Given a plant $G$ over $\Sigma$, a supervisor $V$ on $G$ w.r.t. control constraint $(\Sigma_{c}, \Sigma_{o})$, an attack constraint $(\Sigma_{c, A}, \Sigma_{o, A})$ over $\Sigma$ and a damage-inflicting set $L_{dmg} \subseteq \Sigma^*$, does there exist a successful attacker on $(G, V)$ w.r.t. $(\Sigma_{c, A}, \Sigma_{o, A})$ and $L_{dmg}$? 
\item Given a plant $G$ over $\Sigma$, a control constraint $(\Sigma_{c}, \Sigma_{o})$ over $\Sigma$, a specification $E \subseteq \Sigma^*$, an attack constraint $(\Sigma_{c, A}, \Sigma_{o, A})$ over $\Sigma$ and a damage-inflicting set $L_{dmg} \subseteq \Sigma^*$, does there exist a supervisor $V$ on $G$ w.r.t. $(\Sigma_{c}, \Sigma_{o})$ such that 1) $L(V/G) \subseteq E$ and 2) there is no successful attacker on $(G, V)$ w.r.t. $(\Sigma_{c, A}, \Sigma_{o, A})$ and $L_{dmg}$?
\end{enumerate}
In the rest of this work, we shall present a solution to Problem 1 under the simplifying assumption of normality; Problem 2 will be left as future work. 
\begin{figure}[h]
\centering
\hspace*{-1mm}
\includegraphics[width=2.6in, height=1.8in]{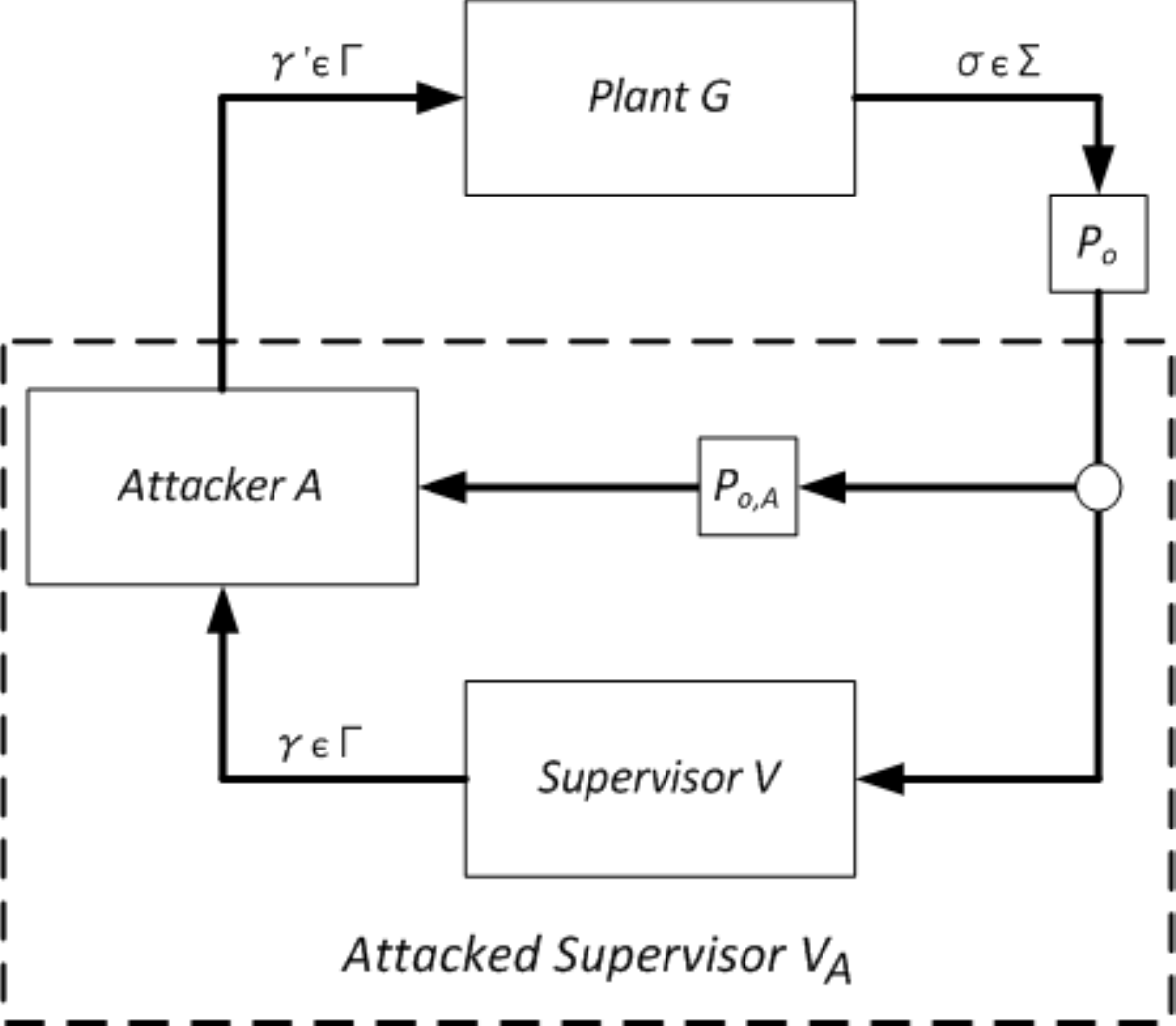} 
\captionsetup{justification=centering}
\caption{The equivalent attacked supervisor $V_A$}
\label{fig:attack}
\end{figure}
\section{Attack Synthesis for Normal Supervisors}
\label{section: ASNS}
A supervisor $V: P_{o}(L(V/G)) \rightarrow \Gamma$ is said to be {\em normal} if it is w.r.t. a control constraint $(\Sigma_{c}, \Sigma_o)$, where $\Sigma_c \subseteq \Sigma_o$. Then, we have that $\Sigma_{uo} \subseteq \Sigma_{uc}$, that is, unobservable events to the supervisor cannot be disabled by the supervisor. Thus, for any $w \in P_{o}(L(V/G))$, we have $V(w) \supseteq \Sigma_{uc} \supseteq \Sigma_{uo}$ by the definition of $\Gamma$. A supervisor $V$ (on $G$) w.r.t. control constraint $(\Sigma_{c}, \Sigma_o)$ is realized by a finite state automaton $S=(X, \Sigma, \zeta, x_0)$ that satisfies the  controllability and observability constraints\footnote{There are different realizations of the same map $V$. We focus on a canonical form of a supervisor proposed in~\cite{B1993}; other realizations of $V$ can be used with no essential difference, apart from the consideration of technical convenience.}~\cite{B1993}:
\begin{enumerate}
\item ({\em controllability}) for any state $x \in X$ and any uncontrollable event $\sigma \in \Sigma_{uc}$, $\zeta(x, \sigma)!$,
\item ({\em observability}) for any state $x \in X$ and any unobservable event $\sigma \in \Sigma_{uo}$, $\zeta(x, \sigma)!$ implies $\zeta(x, \sigma)=x$,
\end{enumerate}
where we have $L(V/G)=L(S \lVert G)$ and, for any $s \in L(G)$ such that $P_o(s) \in P_{o}(L(V/G))$, $V(P_o(s))=\{\sigma \in \Sigma \mid \zeta(\zeta(x_0, s'), \sigma)!, s' \in L(S) \cap P_o^{-1}P_o(s)\}$. For a normal supervisor $S$ w.r.t. control constraint $(\Sigma_{c}, \Sigma_o)$, the observability constraint is then reduced to: for any state $x \in X$ and any unobservable event $\sigma \in \Sigma_{uo}$, $\zeta(x, \sigma)=x$. 

In this section, we investigate the actuator attacker's attack strategy when the supervisors satisfy the normality property. Normal supervisors leave much reduced attack surfaces for the attacker, since events that are disabled must be observable to the supervisor and enabling those disabled (attackable) events (by the actuator attacker) may immediately break the covertness of the attacker, if they are also enabled by the plant. In other words, once the plant generates any string $s \notin L(V/G)$, which is only possible after the attacker enables a disabled attackable event, the supervisor will immediately halt the execution of the closed-loop system. Based on this observation, we shall now show that, with the normality assumption imposed on the supervisors, the definition of $L(V_A/G)$ can be simplified. 
\begin{lemma}
\label{lemma: DOL}
If $V$ is normal on $G$ w.r.t. $(\Sigma_c, \Sigma_o)$, then $L(V_A/G)$ is inductively defined as follows: 
\begin{enumerate}
\item $\epsilon \in L(V_A/G)$,
\item if $s \in L(V_A/G)$, $s \in L(V/G)$, $\sigma \in V_A(P_{o}(s))$ and $s\sigma \in L(G)$,  then $s\sigma \in L(V_A/G)$,
\item no other strings belong to $L(V_A/G)$.
\end{enumerate}
\end{lemma}
\begin{proof}
This follows from $L(G) \cap P_o^{-1}P_o(L(V/G))=L(V/G)$ for supervisors that satisfy the normality condition~\cite{WMW10}. 
\end{proof}
Moreover, the goal of the attacker can also be simplified when $V$ is assumed to be normal.
\begin{lemma}
\label{Lemma: EQ}
If $V$ is normal on $G$ w.r.t. $(\Sigma_c, \Sigma_o)$, then,  
\begin{enumerate}
\item $L(V_A/G) \cap L_{dmg}\neq \varnothing$, and
\item $L(V_A/G)-P_{o}^{-1}P_o(L(V/G))\subseteq L_{dmg}\Sigma^*$
\end{enumerate}
iff $\varnothing \neq L(V_A/G)-L(V/G) \subseteq L_{dmg}$.
\end{lemma}
\begin{proof}
Since $V$ is normal, we have $L(G) \cap P_{o}^{-1}P_o(L(V/G))=L(V/G)$. Thus, $L(V_A/G)-L(V/G)=L(V_A/G)-(L(G) \cap P_{o}^{-1}P_o(L(V/G)))=L(V_A/G)-P_{o}^{-1}P_o(L(V/G))$. It is clear that the right hand side (of the ``iff" statement) implies the left hand side (of the ``iff" statement). In the following, we only need to show that the left hand side implies the right hand side. Suppose the left hand side holds. That is, we have 
\begin{enumerate}
\item $L(V_A/G) \cap L_{dmg}\neq \varnothing$, and
\item $L(V_A/G)-L(V/G)\subseteq L_{dmg}\Sigma^*$ (since $L(V_A/G)-P_{o}^{-1}P_o(L(V/G))=L(V_A/G)-L(V/G)$).
\end{enumerate}
Then, by 1), we know that there exists some string $s \in L(V_A/G) \cap L_{dmg}$. Since $L_{dmg} \cap L(V/G)=\varnothing$, we then have that $s \notin L(V/G)$. Thus, we conclude that $s \in L(V_A/G)-L(V/G)$ and $\varnothing \neq L(V_A/G)-L(V/G)$. Now, let $s$ be any string in $L(V_A/G)-L(V/G)$. By the definition of $L(V_A/G)$ and Lemma~\ref{lemma: DOL}, we know\footnote{$s$ cannot be $\epsilon$ since $\epsilon \in L(V/G)$.} that $s=s'\sigma$ for some $s' \in L(V/G)$. Since $s' \in L(V/G)$, we have $\overline{s'} \cap L_{dmg}=\varnothing$. Since $s'\sigma=s \in L_{dmg}\Sigma^*$ by 2), we have that $\overline{s'\sigma} \cap L_{dmg} \neq \varnothing$. The only possibility is that $s=s'\sigma \in L_{dmg}$. Thus, we have $L(V_A/G)-L(V/G)\subseteq L_{dmg}$. In conclusion, we have that $\varnothing \neq L(V_A/G)-L(V/G) \subseteq L_{dmg}$.
\end{proof}

\begin{remark}
For a normal supervisor $V$, by Lemma~\ref{Lemma: EQ}, $A$ is successful on $(G, V)$ w.r.t. $(\Sigma_{c, A}, \Sigma_{o, A})$ and $L_{dmg}$ iff $\varnothing \neq L(V_A/G)-L(V/G) \subseteq L_{dmg}$.  
\end{remark}
We make the following simplifying assumption throughout the rest of this work.
\begin{assumption}
\label{assump: normalattack}
$\Sigma_{c, A} \subseteq \Sigma_{o, A}$ holds for the attack constraint $(\Sigma_{c, A}, \Sigma_{o, A})$. \qed
\end{assumption}
\begin{remark}
\label{remark: asspt}
Intuitively, Assumption~\ref{assump: normalattack} states that the attacker cannot attack events that it cannot observe. With Assumption~\ref{assump: normalattack}, we then have the following inclusion relations in the case of attacking normal supervisors.
\begin{enumerate}
\item $\Sigma_{c, A} \subseteq \Sigma_{o, A} \subseteq \Sigma_o$
\item $\Sigma_{c, A} \subseteq \Sigma_{c} \subseteq \Sigma_o$.
\end{enumerate}
\end{remark}
An attacker $A$ is said to be {\em normal} if it is w.r.t. an attack constraint $(\Sigma_{c, A}, \Sigma_{o, A})$, where $\Sigma_{c, A} \subseteq \Sigma_{o, A}$. Thus, in this work we consider the special case of normal supervisors and normal attackers. We shall now introduce the notion of an enabling actuator attacker, which plays an important role in solving the problem of synthesizing attackers in this setup (see Remark~\ref{remark: asspt}). 
\begin{definition}
An attacker $A$ on $(G, V)$ w.r.t. $(\Sigma_{c, A}, \Sigma_{o, A})$ is said to be an {\em enabling} attacker if $A(\hat{P}_{o, A}^{V}(w)) \supseteq V(w) \cap \Sigma_{c, A}$, for any $w \in P_o(L(V/G))$. \qed
\end{definition}
For an enabling actuator attacker $A$, we have that $V_A(w)=(V(w)-\Sigma_{c, A}) \cup A(\hat{P}_{o, A}^{V}(w)) \supseteq V(w)$, for any $w \in P_o(L(V/G))$. Intuitively, it never disables attackable events that are enabled by the supervisor, that is, the only attack decision of an enabling attacker is to enable attackable events that are disabled by the supervisor. Essentially, an enabling actuator attacker only performs AE attacks~\cite{CarvalhoEnablementAttacks}. It follows that, for any enabling actuator attacker $A$, we have $L(V_A/G) \supseteq L(V/G)$. We are now ready to state the following theorem; it states that, with the normality assumption imposed on the supervisors and attackers, we only need to focus on AE attacks in order to determine the existence of successful attackers. 
\begin{theorem}
\label{theorem:existence}
Let the plant $G$ and the normal supervisor $V$ be given. Then, the following three statements are equivalent. 
\begin{enumerate}
\item There exists a successful attacker on $(G, V)$ w.r.t. $(\Sigma_{c, A}, \Sigma_{o, A})$ and $L_{dmg}$.
\item There exists a successful enabling attacker on $(G, V)$ w.r.t. $(\Sigma_{c, A}, \Sigma_{o, A})$ and $L_{dmg}$.
\item There exists a string $s \in L(V/G)$ and an attackable event $\sigma \in \Sigma_{c, A}$, such that 
\begin{enumerate}
\item $s\sigma \in L_{dmg}$
\item for any $s' \in L(V/G)$, $\hat{P}_{o, A}^{V}(P_o(s))=\hat{P}_{o, A}^{V}(P_o(s'))$ and $s'\sigma \in L(G)$ together implies $s'\sigma \in L_{dmg}$. \qed
\end{enumerate}
\end{enumerate}
\end{theorem}
\begin{proof}
Clearly, 2) implies 1). We shall first show 1) implies 3).

Suppose 1) holds. Let $A$ be a successful attacker on $(G, V)$ w.r.t. $(\Sigma_{c, A}, \Sigma_{o, A})$ and $L_{dmg}$. By definition, we have $\varnothing \neq L(V_A/G)-L(V/G) \subseteq L_{dmg}$. Let us now fix some string $s_0$ in $L(V_A/G)-L(V/G)$, the existence of which is ensured since $L(V_A/G)-L(V/G) \neq \varnothing$. It is not difficult to see that $s_0=s\sigma$ for some $s \in L(V/G) \cap L(V_A/G)$ and $\sigma \in  A(\hat{P}_{o, A}^{V}(P_o(s)))-V(P_o(s)) \subseteq \Sigma_{c, A}$. Since $s\sigma=s_0 \in L(V_A/G)-L(V/G) \subseteq L_{dmg}$, we conclude that 3. a) holds. In the following, we shall show that 3. b) also holds.

Let $s'$ be any string in $L(V/G)$. Suppose that $\hat{P}_{o, A}^{V}(P_o(s'))=\hat{P}_{o, A}^{V}(P_o(s))$ and $s'\sigma \in L(G)$. We need to show that $s'\sigma \in L_{dmg}$. 

 It is not difficult to see that $\hat{P}_{o, A}^{V}(P_o(s'))=\hat{P}_{o, A}^{V}(P_o(s))$ implies $V(P_o(s'))=V(P_o(s))$. Indeed, $\hat{P}_{o, A}^{V}(P_o(s'))=\hat{P}_{o, A}^{V}(P_o(s))$ implies $last(\hat{P}_{o, A}^{V}(P_o(s')))=last(\hat{P}_{o, A}^{V}(P_o(s)))$, where $last$ denotes the operation that computes the last element in a string, and by the definition of $\hat{P}_{o, A}^{V}$ we have 
\begin{center}
$last(\hat{P}_{o, A}^{V}(P_o(s')))=(P_{o, A}(last(P_o(s'))), V(P_o(s')))$ and $last(\hat{P}_{o, A}^{V}(P_o(s)))=(P_{o, A}(last(P_o(s))), V(P_o(s)))$
\end{center}
We shall only need to show that $s' \in L(V_A/G)$. Then, it follows that $s'\sigma \in L(V_A/G)-L(V/A) \subseteq L_{dmg}$, since $s'\sigma \in L(G)$ and $\sigma \in A(\hat{P}_{o, A}^{V}(P_o(s)))=A(\hat{P}_{o, A}^{V}(P_o(s'))) \subseteq V_A(P_o(s'))$ and $\sigma \notin V(P_o(s))=V(P_o(s'))$. 

Suppose, on the contrary, that $s' \notin L(V_A/G)$, then it must be the case that $s'=s_1'\sigma's_2'$ for some $s_1', s_2' \in \Sigma^*$ and $\sigma' \in \Sigma$ such that $s_1' \in L(V/G) \cap L(V_A/G)$ and $s_1'\sigma' \in L(V/G)-L(V_A/G)$. Thus, we have $\sigma' \in V(P_o(s_1'))$ and $\sigma' \notin V_A(P_o(s_1'))=(V(P_o(s_1'))-\Sigma_{c, A}) \cup  A(\hat{P}_{o, A}^{V}(P_o(s_1')))$. It follows that $\sigma' \in V(P_o(s_1')) \cap \Sigma_{c, A}-A(\hat{P}_{o, A}^{V}(P_o(s_1'))) \subseteq \Sigma_{c, A}$. Then, by Assumption~\ref{assump: normalattack}, we have $\sigma'\in \Sigma_{o, A}$. It is immediate that $\hat{P}_{o, A}^{V}(P_o(s_1'))$ is a strict prefix of $\hat{P}_{o, A}^{V}(P_o(s'))=\hat{P}_{o, A}^{V}(P_o(s))$, since $s_1'\sigma' \leq s'$ and $\sigma' \in \Sigma_{o, A}$. It follows that the set
\begin{center}
$D=\{s_1 \in \Sigma^* \mid s_1\leq s, \hat{P}_{o, A}^{V}(P_o(s_1))=\hat{P}_{o, A}^{V}(P_o(s_1'))\}$
\end{center} is non-empty. Let $s^0_{1}$ be the longest string in $D$. Then, $\hat{P}_{o, A}^{V}(P_o(s^0_1))=\hat{P}_{o, A}^{V}(P_o(s_1'))$ is a strict prefix of $\hat{P}_{o, A}^{V}(P_o(s))$ and it follows that $s\neq s^0_1$. Since $s^0_1 \leq s$ by the definition of $D$, we must have $s^0_{1}<s$. It follows that $s=s^0_{1}\sigma''s^0_2$ for some $\sigma'' \in \Sigma$ and $s^0_2 \in \Sigma^*$. We first observe that $\sigma'' \in \Sigma_o$, since otherwise $s^0_1\sigma''$ would be another string in $D$ and it is longer than $s^0_1$, which contradicts the fact that $s^0_1$ is the longest string in $D$. Next, we show that $\sigma'' \in \Sigma_{o, A}$. Otherwise, we have $\sigma'' \in \Sigma_o-\Sigma_{o, A}$ and  the following.
\begin{enumerate}
\item [i)] $\hat{P}_{o, A}^{V}(P_o(s^0_1))=\hat{P}_{o, A}^{V}(P_o(s_1'))$
\item [ii)] $\hat{P}_{o, A}^{V}(P_o(s^0_1\sigma''))=\hat{P}_{o, A}^{V}(P_o(s^0_1))(\epsilon, V(P_o(s^0_1\sigma'')))$
\item [iii)] $\hat{P}_{o, A}^{V}(P_o(s_1'\sigma'))=\hat{P}_{o, A}^{V}(P_o(s_1))(\sigma', V(P_o(s_1'\sigma')))$
\end{enumerate}
It follows that $\hat{P}_{o, A}^{V}(P_o(s^0_1\sigma'')) \neq \hat{P}_{o, A}^{V}(P_o(s_1'\sigma'))$ and $|\hat{P}_{o, A}^{V}(P_o(s^0_1\sigma''))|=|\hat{P}_{o, A}^{V}(P_o(s_1'\sigma'))|$. Since $s^0_1\sigma''$ is a prefix of $s$ and $s_1'\sigma'$ is a prefix of $s'$, we conclude that $\hat{P}_{o, A}^{V}(P_o(s^0_1\sigma''))$ is a prefix of $\hat{P}_{o, A}^{V}(P_o(s))$ and $\hat{P}_{o, A}^{V}(P_o(s_1'\sigma'))$ is a prefix of $\hat{P}_{o, A}^{V}(P_o(s'))$; then, the prefix of $\hat{P}_{o, A}^{V}(P_o(s))$ of length $|\hat{P}_{o, A}^{V}(P_o(s^0_1\sigma''))|$ is different from the prefix of $\hat{P}_{o, A}^{V}(P_o(s'))$ of the same length. However, this contradicts the supposition that $\hat{P}_{o, A}^{V}(P_o(s'))=\hat{P}_{o, A}^{V}(P_o(s))$. That is, we must have $\sigma'' \in \Sigma_{o, A}$. Moreover, we must have $\sigma''=\sigma'$; otherwise, the same contradiction can be reached by comparing $\hat{P}_{o, A}^{V}(P_o(s^0_1\sigma''))$ and $\hat{P}_{o, A}^{V}(P_o(s_1'\sigma'))$. Since $\hat{P}_{o, A}^{V}(P_o(s^0_1))=\hat{P}_{o, A}^{V}(P_o(s_1'))$, we have that $\sigma' \notin A(\hat{P}_{o, A}^{V}(P_o(s_1')))=A(\hat{P}_{o, A}^{V}(P_o(s^0_1)))$. Thus, we have $\sigma'' \notin A(\hat{P}_{o, A}^{V}(P_o(s^0_1)))$ and $\sigma''=\sigma' \in \Sigma_{c, A}$. Then, we can conclude that $\sigma'' \notin V_A(P_o(s^0_1))$. However, this contradicts with the fact that $s^0_1\sigma'' \in L(V_A/G)$. Thus, we conclude that indeed $s' \in L(V_A/G)$. It follows that 3. b) holds. 

Finally, we shall show that 3) implies 2). 

Now, suppose 3) holds. Let $s \in L(V/G)$ and $\sigma \in \Sigma_{c, A}$ such that 1) $s\sigma \in L_{dmg}$, and 2) for any string $s' \in L(V/G)$, $\hat{P}_{o, A}^{V}(P_o(s))=\hat{P}_{o, A}^{V}(P_o(s'))$ and $s'\sigma \in L(G)$ together implies $s'\sigma \in L_{dmg}$. It is clear that $\sigma \notin V(P_o(s))$. Let $A: \hat{P}_{o, A}^{V}(P_o(L(V/G))) \rightarrow \Delta \Gamma$ be an attacker on $(G, V)$ w.r.t. $(\Sigma_{c, A}, \Sigma_{o, A})$, where 
\begin{enumerate}
\item $A(\hat{P}_{o, A}^{V}(P_o(s)))=(V(P_o(s)) \cap \Sigma_{c, A}) \cup \{\sigma\}$ 
\item for any $w \in P_o(L(V/G))$ such that $\hat{P}_{o, A}^{V}(w)\neq \hat{P}_{o, A}^{V}(P_o(s))$, $A(\hat{P}_{o, A}^{V}(w))=V(w) \cap \Sigma_{c, A}$. 
\end{enumerate}
Intuitively, the attacker $A$ performs no actual attack until it observes $\hat{P}_{o, A}^{V}(P_o(s))$, at which point the attacker enables the disabled attackable event $\sigma$. It is immediate that $A$ is a well-defined actuator attacker, since for any $w \in P_o(L(V/G))$ we have that $V(w)$ is the second component of $last(\hat{P}_{o, A}^{V}(w))$. That is, for any $w_1, w_2 \in P_o(L(V/G))$ such that $\hat{P}_{o, A}^{V}(w_1)=\hat{P}_{o, A}^{V}(w_2) \neq \hat{P}_{o, A}^{V}(P_o(s))$, we have that $A(\hat{P}_{o, A}^{V}(w_1))=A(\hat{P}_{o, A}^{V}(w_2))$, since $V(w_1)=V(w_2)$. It is clear that $A$ is an enabling actuator attacker by definition. In the rest, we show that $\varnothing \neq L(V_A/G)-L(V/G) \subseteq L_{dmg}$, that is, $A$ is a successful attacker on $(G, V)$ w.r.t. $(\Sigma_{c, A}, \Sigma_{o, A})$ and $L_{dmg}$.


Since $A$ is an enabling attacker, we have $L(V/G) \subseteq L(V_A/G)$. Since $s \in L(V/G)$ and $\sigma \in A(\hat{P}_{o, A}^{V}(P_o(s))) \subseteq V_A(P_o(s))$ and $s\sigma \in L_{dmg} \subseteq L(G)$, we have $s\sigma \in L(V_A/G)$ and $s\sigma \notin L(V/G)$, i.e., $s\sigma \in L(V_A/G)-L(V/G)$. We conclude that $L(V_A/G)-L(V/G)\neq \varnothing$. In the rest, we only need to show that $L(V_A/G)-L(V/G) \subseteq L_{dmg}$. Let $s''$ be any string in $L(V_A/G)-L(V/G)$. By the definition of $L(V_A/G)$ and Lemma~\ref{lemma: DOL}, we conclude that $s''=s'\sigma'$ for some $s' \in L(V/G)$ and some $\sigma' \in \Sigma_{c, A}$, where we have $\sigma' \notin V(P_o(s'))$ and $\sigma' \in A(\hat{P}_{o, A}^{V}(P_o(s'))) \subseteq V_A(P_o(s'))$. We then conclude that $V(P_o(s')) \neq V_A(P_o(s'))$ and also $\sigma' \in V_A(P_o(s'))-V(P_o(s'))$. Now, from the definition of $A$, we observe that, for any $w \in P_o(L(V/G))$, $V_A(w) \neq V(w)$ only if\footnote{$V_A(w)=(V(w)-\Sigma_{c, A}) \cup A(\hat{P}_{o, A}^{V}(w)) \neq V(w)$ only if $A(\hat{P}_{o, A}^{V}(w)) \neq V(w) \cap \Sigma_{c, A}$.} $A(\hat{P}_{o, A}^{V}(w)) \neq V(w) \cap \Sigma_{c, A}$ and $\hat{P}_{o, A}^{V}(w)=\hat{P}_{o, A}^{V}(P_o(s))$ and $V_A(w)=V(w) \cup \{\sigma\}$.

 Thus, we conclude that $V_A(P_o(s'))=V(P_o(s')) \cup \{\sigma\}$ and $\hat{P}_{o, A}^{V}(P_o(s'))=\hat{P}_{o, A}^{V}(P_o(s))$; the only possibility is $\sigma'=\sigma$. Thus, we have that $s' \in L(V/G)$, $\hat{P}_{o, A}^{V}(P_o(s))=\hat{P}_{o, A}^{V}(P_o(s'))$ and $s'\sigma=s'\sigma'=s'' \in L(G)$, which implies $s''=s'\sigma \in L_{dmg}$ by the supposition. That is, we have shown that $L(V_A/G)-L(V/G) \subseteq L_{dmg}$. That is, 2) holds.
\end{proof}
Based on Theorem~\ref{theorem:existence}, we are able to identify a notion of attackability and a notion of an attack pair for ensuring the existence of a successful (enabling) attacker, in the case of normal supervisors.
\begin{definition} 
Let the plant $G$ and the normal supervisor $V$ be given. $(G, V)$ is said to be {\em attackable} w.r.t. $(\Sigma_{c, A}, \Sigma_{o, A})$ and $L_{dmg}$ if there exists a string $s \in L(V/G)$ and an attackable event $\sigma \in \Sigma_{c, A}$, such that 
\begin{enumerate}
\item $s\sigma \in L_{dmg}$
\item for any $s' \in L(V/G)$, $\hat{P}_{o, A}^{V}(P_o(s))=\hat{P}_{o, A}^{V}(P_o(s'))$ and $s'\sigma \in L(G)$ together implies $s'\sigma \in L_{dmg}$. 
\end{enumerate}
Moreover, each pair $(s, \sigma) \in L(V/G) \times \Sigma_{c, A}$ that satisfies Conditions 1) and 2) is said to be an {\em attack pair} for $(G, V)$ w.r.t. $(\Sigma_{c, A}, \Sigma_{o, A})$ and $L_{dmg}$. \qed
\end{definition}
Thus, Theorem~\ref{theorem:existence} states that there exists a successful (enabling) attacker on $(G, V)$ w.r.t. $(\Sigma_{c, A}, \Sigma_{o, A})$ and $L_{dmg}$ iff $(G, V)$ is attackable w.r.t. $(\Sigma_{c, A}, \Sigma_{o, A})$ and $L_{dmg}$. Furthermore, based on the proof of Theorem~\ref{theorem:existence}, a successful (enabling) attacker $A$ on $(G, V)$ w.r.t. $(\Sigma_{c, A}, \Sigma_{o, A})$ and $L_{dmg}$ can be readily synthesized if an attack pair for $(G, V)$ w.r.t. $(\Sigma_{c, A}, \Sigma_{o, A})$ and $L_{dmg}$ can be obtained. We now make the following remark.  
\begin{remark}
Given any two supervisors $V_1, V_2$ on $G$ w.r.t. $(\Sigma_c, \Sigma_o)$, $(G, V_1)$ and $(G, V_2)$ are said to be {\em equi-attackable} w.r.t. $(\Sigma_{c, A}, \Sigma_{o, A})$ and $L_{dmg}$ if $(G, V_1)$ is attackable w.r.t. $(\Sigma_{c, A}, \Sigma_{o, A})$ and $L_{dmg}$ iff $(G, V_2)$ is attackable w.r.t. $(\Sigma_{c, A}, \Sigma_{o, A})$ and $L_{dmg}$. It is worth noting that $L(V_1/G)=L(V_2/G)$ does not imply the equi-attackability of $(G, V_1)$ and $(G, V_2)$, since in general $V_1 \neq V_2$ and thus $\hat{P}_{o, A}^{V_1} \neq \hat{P}_{o, A}^{V_2}$.

\end{remark}
\subsection{Characterization of the Supremal Successful Enabling Attacker}
\label{sub: CSSEA}
From now on, we shall mainly focus on the class of enabling attackers on $(G, V)$ w.r.t. $(\Sigma_{c, A}, \Sigma_{o, A})$. It is of interest to obtain a successful enabling attacker $A$ that generates as many damaging strings in $L_{dmg}$ as possible in $L(V_A/G)$. Now, let $AP_{dmg}$ denote the set of all the attack pairs for $(G, V)$ w.r.t. $(\Sigma_{c, A}, \Sigma_{o, A})$ and $L_{dmg}$. For any $(s, \sigma) \in AP_{dmg}$, we define $En(s, \sigma):=\{s'\sigma \in L(G) \mid \hat{P}_{o, A}^{V}(P_o(s))=\hat{P}_{o, A}^{V}(P_o(s')), s' \in L(V/G)\}$. Intuitively, $En(s, \sigma)$ denotes the set of strings in $L(G)$ that are one-step $\sigma$-extensions $s'\sigma$ of strings $s'$ in $L(V/G)$, where $s'$ and $s$ are attacker observation equivalent strings, i.e., $\hat{P}_{o, A}^{V}(P_o(s))=\hat{P}_{o, A}^{V}(P_o(s'))$. 
It is clear that $s\sigma \in En(s, \sigma) \subseteq L_{dmg}$ holds. We now have the following characterization of attacked closed-behavior under successful enabling attackers. 
\begin{proposition}
\label{proposition: char}
Let the plant $G$ and the normal supervisor $V$ be given. Suppose $AP_{dmg} \neq \varnothing$. Then, for any enabling attacker $A$ on $(G, V)$ w.r.t. $(\Sigma_{c, A}, \Sigma_{o, A})$, $A$ is successful on $(G, V)$ w.r.t. $(\Sigma_{c, A}, \Sigma_{o, A})$ and $L_{dmg}$ iff $L(V_A/G)=L(V/G) \cup \bigcup_{(s, \sigma) \in B}En(s, \sigma)$ for some non-empty $B \subseteq AP_{dmg}$.
\end{proposition}
\begin{proof}
Suppose $AP_{dmg} \neq \varnothing$. Then, we know that $(G, V)$ is attackable w.r.t. $(\Sigma_{c, A}, \Sigma_{o, A})$ and $L_{dmg}$. Let $A$ be any enabling attacker on $(G, V)$ w.r.t. $(\Sigma_{c, A}, \Sigma_{o, A})$. We have $L(V/G) \subseteq L(V_A/G)$.

Suppose $L(V_A/G)=L(V/G) \cup \bigcup_{(s, \sigma) \in B}En(s, \sigma)$ for some non-empty $B \subseteq AP_{dmg}$. Let $(s_0, \sigma_0) \in B \subseteq AP_{dmg}$ be an arbitrarily chosen attack pair on $(G, V)$ w.r.t. $(\Sigma_{c, A}, \Sigma_{o, A})$ and $L_{dmg}$, the existence of which is ensured since $B$ is non-empty. We have $s_0\sigma_0 \in L_{dmg}$ and $s_0\sigma_0 \in En(s_0, \sigma_0) \subseteq L(V_A/G)$, by the supposition that $L(V_A/G)=L(V/G) \cup \bigcup_{(s, \sigma) \in B}En(s, \sigma)$. Thus, $s_0\sigma_0 \in L(V_A/G)-L(V/G)$ and $L(V_A/G)-L(V/G) \neq \varnothing$. We have $L(V_A/G)=L(V/G) \cup \bigcup_{(s, \sigma) \in B}En(s, \sigma) \subseteq L(V/G) \cup L_{dmg}$. Thus, $L(V_A/G)-L(V/G) \subseteq L_{dmg}$. Thus, $A$ is successful on $(G, V)$ w.r.t. $(\Sigma_{c, A}, \Sigma_{o, A})$ and $L_{dmg}$.

Now, suppose $A$ is successful on $(G, V)$ w.r.t. $(\Sigma_{c, A}, \Sigma_{o, A})$ and $L_{dmg}$, that is, $\varnothing \neq L(V_A/G)-L(V/G) \subseteq L_{dmg}$. Let $s_0$ be any string in $L(V_A/G)-L(V/G)$. It is clear that $s_0=s\sigma$ for some $s \in L(V/G)$ and $\sigma \in \Sigma_{c, A}$. From the proof of Theorem~\ref{theorem:existence}, $(s, \sigma)$ is an attack pair for $(G, V)$ w.r.t. $(\Sigma_{c, A}, \Sigma_{o, A})$ and $L_{dmg}$. Now, Let $B=\{(s, \sigma) \in L(V/G) \times \Sigma_{c, A} \mid s\sigma \in L(V_A/G)-L(V/G)\}$. It is clear that $\varnothing \neq B \subseteq AP_{dmg}$. We now show that $L(V_A/G)-L(V/G)=\bigcup_{(s, \sigma) \in B}En(s, \sigma)$. It then follows from $L(V/G) \subseteq L(V_A/G)$ that $L(V_A/G)=L(V/G) \cup \bigcup_{(s, \sigma) \in B}En(s, \sigma)$.

It is clear that $L(V_A/G)-L(V/G) \subseteq \bigcup_{(s, \sigma) \in B}En(s, \sigma)$, since, for any $s_0=s\sigma \in L(V_A/G)-L(V/G)$, we have $(s, \sigma) \in B$ and $s_0=s\sigma \in En(s, \sigma) \subseteq \bigcup_{(s, \sigma) \in B}En(s, \sigma)$.

On the other hand, let $s_0$ be any string in $\bigcup_{(s, \sigma) \in B}En(s, \sigma)$. By definition, there exists some $(s, \sigma) \in B$ such that $s_0 \in En(s, \sigma)$. We have $s\sigma \in L(V_A/G)-L(V/G)$ and $s_0=s'\sigma \in L(G)$ for some $s' \in L(V/G)$ such that $\hat{P}_{o, A}^{V}(P_o(s))=\hat{P}_{o, A}^{V}(P_o(s'))$. It is clear that $\sigma \in A(\hat{P}_{o, A}^{V}(P_o(s)))=A(\hat{P}_{o, A}^{V}(P_o(s'))) \subseteq V_A(P_o(s'))$ and thus $s_0=s'\sigma \in L(V_A/G)$. $\hat{P}_{o, A}^{V}(P_o(s))=\hat{P}_{o, A}^{V}(P_o(s'))$ implies $V(P_o(s))=V(P_o(s'))$. From $\sigma \notin V(P_o(s))$, we conclude that $\sigma \notin V(P_o(s'))$ and thus $s_0=s'\sigma \notin L(V/G)$. Then, we have $s_0 \in L(V_A/G)-L(V/G)$. That is, $\bigcup_{(s, \sigma) \in B}En(s, \sigma) \subseteq L(V_A/G)-L(V/G)$. 
\end{proof}
It follows that the largest possible attacked closed-behavior under successful enabling attackers is $L(V/G) \cup \bigcup_{(s, \sigma) \in AP_{dmg}}En(s, \sigma)$, by Proposition~\ref{proposition: char}, if $AP_{dmg}$ is non-empty. Now, we show that it is realized by the enabling attacker $A^{sup}: \hat{P}_{o, A}^{V}(P_o(L(V/G))) \rightarrow \Delta \Gamma$ such that $A^{sup}(\hat{P}_{o, A}^{V}(P_o(s)))=(V(P_o(s)) \cap \Sigma_{c, A}) \cup I(\hat{P}_{o, A}^{V}(P_o(s)))$, for any $s \in L(V/G)$, where $I(\hat{P}_{o, A}^{V}(P_o(s))):=\{\sigma \in \Sigma_{c, A} \mid \exists s' \in L(V/G), \hat{P}_{o, A}^{V}(P_o(s))=\hat{P}_{o, A}^{V}(P_o(s')) \wedge (s', \sigma) \in AP_{dmg}\}$. It is clear that $I(\hat{P}_{o, A}^{V}(P_o(s))) \cap V(P_o(s))=\varnothing$ by the definition of $I(\hat{P}_{o, A}^{V}(P_o(s)))$. Intuitively, after the closed-loop system generates a string $s \in L(V/G) \subseteq L(V_{A^{sup}}/G)$, the attacker $A^{sup}$ enables an attackable event\footnote{The attackable events in $V(P_o(s)) \cap \Sigma_{c, A}$ are always enabled by $A^{sup}$ since $A^{sup}$ is enabling.} $\sigma \in \Sigma_{c, A}-V(P_o(s))$ iff there exists a string $s' \in L(V/G)$ that is observationally equivalent to $s$ from the attacker's point of view, i.e., $\hat{P}_{o, A}^{V}(P_o(s))=\hat{P}_{o, A}^{V}(P_o(s'))$, such that $(s', \sigma)$ is an attack pair. We first show that $A^{sup}$ is a well-defined attacker. That is, for any two strings $s, s'' \in L(V/G)$ such that $\hat{P}_{o, A}^{V}(P_o(s))=\hat{P}_{o, A}^{V}(P_o(s''))$, we need to show that $A^{sup}(\hat{P}_{o, A}^{V}(P_o(s)))=A^{sup}(\hat{P}_{o, A}^{V}(P_o(s'')))$. 

\begin{lemma}
$A^{sup}$ is a well-defined attacker on $(G, V)$ w.r.t. $(\Sigma_{c, A}, \Sigma_{o, A})$. 
\end{lemma}
\begin{proof}

Since $\hat{P}_{o, A}^{V}(P_o(s))=\hat{P}_{o, A}^{V}(P_o(s''))$ implies that  $V(P_o(s))=V(P_o(s''))$, we only need to show that $I(\hat{P}_{o, A}^{V}(P_o(s)))=I(\hat{P}_{o, A}^{V}(P_o(s'')))$. This is straightforward from $\hat{P}_{o, A}^{V}(P_o(s))=\hat{P}_{o, A}^{V}(P_o(s''))$ and the definitions of $I(\hat{P}_{o, A}^{V}(P_o(s)))$, $I(\hat{P}_{o, A}^{V}(P_o(s'')))$. 
\end{proof}
\begin{lemma}
\label{lemma: supre}
Let the plant $G$ and the normal supervisor $V$ be given. Suppose $AP_{dmg} \neq \varnothing$. Then, $L(V_{A^{sup}}/G)=L(V/G) \cup \bigcup_{(s, \sigma) \in AP_{dmg}}En(s, \sigma)$.
\end{lemma}
\begin{proof}
It is clear that $L(V/G) \subseteq L(V_{A^{sup}}/G)$, since $A^{sup}$ is enabling. We need to show $L(V_{A^{sup}}/G)-L(V/G)=\bigcup_{(s, \sigma) \in AP_{dmg}}En(s, \sigma)$.

Let $s''$ be any string in $L(V_{A^{sup}}/G)-L(V/G)$. It is clear that $s''=s'\sigma$ for some $s' \in L(V/G) \cap L(V_{A^{sup}}/G)$ and $\sigma \in \Sigma_{c, A}$; moreover, $\sigma \notin V(P_o(s'))$ and $\sigma \in A^{sup}(\hat{P}_{o, A}^{V}(P_o(s'))) \subseteq V_{A^{sup}}(P_o(s'))$. Thus, we have $V(P_o(s'))\neq V_{A^{sup}}(P_o(s'))$ and $\sigma \in V_{A^{sup}}(P_o(s'))-V(P_o(s'))$. Now, from the definition of $A^{sup}$, we can conclude that $\sigma \in I(\hat{P}_{o, A}^{V}(P_o(s'))) \neq \varnothing$. Then, there exists some $s \in L(V/G)$ such that $\hat{P}_{o, A}^{V}(P_o(s'))=\hat{P}_{o, A}^{V}(P_o(s))$ and $(s, \sigma) \in L_{dmg}$. It is clear that $s'\sigma \in En(s, \sigma)$, by the definition of $En(s, \sigma)$. Then, we have $s''=s'\sigma \in En(s, \sigma) \subseteq \bigcup_{(s, \sigma) \in AP_{dmg}}En(s, \sigma)$.

Let $s'' \in \bigcup_{(s, \sigma) \in AP_{dmg}}En(s, \sigma)$. Then, there exists some $(s, \sigma) \in AP_{dmg}$ such that $s'' \in En(s, \sigma) \subseteq L_{dmg}$. Then, we have that $s''=s'\sigma$ for some $s' \in L(V/G)$ such that $\hat{P}_{o, A}^{V}(P_o(s'))=\hat{P}_{o, A}^{V}(P_o(s))$, from the definition of $En(s, \sigma)$. Since $(s, \sigma) \in AP_{dmg}$, by the definition of $A^{sup}$, we have $\sigma \in I(\hat{P}_{o, A}^{V}(P_o(s')))\subseteq A^{sup}(\hat{P}_{o, A}^{V}(P_o(s'))) \subseteq V_{A^{sup}}(P_o(s'))$. Thus, we have $s''=s'\sigma \in L(V_{A^{sup}}/G)-L(V/G)$. Then, $\bigcup_{(s, \sigma) \in AP_{dmg}}En(s, \sigma) \subseteq L(V_{A^{sup}}/G)-L(V/G)$.
\end{proof}
From Proposition~\ref{proposition: char} and Lemma~\ref{lemma: supre}, we know that $A^{sup}$ generates the largest possible attacked closed-behavior among the set of successful enabling attackers; for this reason, we refer to $A^{sup}$ as the supremal successful enabling attacker\footnote{However, there could be different attackers that also realize $L(V_{A^{sup}}/G)$.}. We now show that the attacked closed-behavior under $A^{sup}$ is the largest possible, even if we take into consideration those attackers that are not enabling. In other words, there is no loss of permissiveness when we focus on the class of enabling attackers. 
\begin{proposition}
Let the plant $G$ and the normal supervisor $V$ be given. Suppose $AP_{dmg} \neq \varnothing$. Then, for any successful attacker $A$ on $(G, V)$ w.r.t. $(\Sigma_{c, A}, \Sigma_{o, A})$ and $L_{dmg}$, it holds that $L(V_A/G) \subseteq L(V_{A^{sup}}/G)$.
\end{proposition}
\begin{proof}
Let $A$ be any successful attacker on $(G, V)$ w.r.t. $(\Sigma_{c, A}, \Sigma_{o, A})$ and $L_{dmg}$, we have $\varnothing \neq L(V_A/G)-L(V/G) \subseteq L_{dmg}$. By Lemma~\ref{lemma: supre}, we only need to show that $L(V_A/G) \subseteq L(V/G) \cup \bigcup_{(s, \sigma) \in AP_{dmg}}En(s, \sigma)$. 

Let $s''$ be any string in $L(V_A/G)$. If $s'' \in L(V/G)$, then $s'' \in L(V/G) \cup \bigcup_{(s, \sigma) \in AP_{dmg}}En(s, \sigma)$. If $s'' \notin L(V/G)$, then we have $s'' \in L(V_A/G)-L(V/G) \subseteq L_{dmg}$; we only need to show that $s'' \in \bigcup_{(s, \sigma) \in AP_{dmg}}En(s, \sigma)$. It is clear that $s''=s\sigma \in L_{dmg}$ for some $s \in L(V/G)$ and some $\sigma \in \Sigma_{c, A}$. From the proof of Theorem~\ref{theorem:existence}, we know that $(s, \sigma)$ is an attack pair. Then, it follows that $s''=s\sigma\in En(s, \sigma) \subseteq \bigcup_{(s, \sigma) \in AP_{dmg}}En(s, \sigma)$.
\end{proof}
\subsection{Synthesis of the Supremal Successful Attacker}
\label{subsection:SSSEA}
In this subsection, we shall address the problem of synthesis of the supremal successful attacker $A^{sup}$ under the normality assumption. 

Recall that $G=(Q, \Sigma, \delta, q_0)$ is the plant, $V$ is the supervisor on $G$ w.r.t. $(\Sigma_c, \Sigma_o)$, which is given by the finite state automaton $S=(X, \Sigma, \zeta, x_0)$, and $L_{dmg}$ is the set of damaging strings, which is a regular language over $\Sigma$. Let $L_{dmg}$ be recognized by the damage automaton $H=(Z, \Sigma, \eta, z_0, Z_m)$, i.e., $L_m(H)=L_{dmg}$\footnote{$L_m(H):=\{s \in L(H) \mid \eta(z_0, s) \in Z_m\}$ is the marked behavior of $H$.}. We require $H$ to be a complete finite state automaton, i.e., $L(H)=\Sigma^*$. We shall now provide a high-level idea of the synthesis algorithm using $G$, $S$ and $H$. 


In order to synthesize the supremal successful attacker $A^{sup}$ on $(G, S)$ w.r.t. $(\Sigma_{c, A}, \Sigma_{o, A})$ and $L_{dmg}$, which is given in Section~\ref{sub: CSSEA}, we need to know the attacker's observation sequence $\hat{P}_{o, A}^{V}(P_o(s))$ after the plant executes a string $s \in L(S \lVert G)$ and also record the set of all the strings $s' \in L(S \lVert G)$ that are observationally equivalent to $s$ from the attacker's point of view, i.e., $\hat{P}_{o, A}^{V}(P_o(s))=\hat{P}_{o, A}^{V}(P_o(s'))$; then, to determine whether an attack on the attackable event $\sigma \in \Sigma_{c, A}$ can be successfully established upon observing $\hat{P}_{o, A}^{V}(P_o(s))$, this set of strings $s'$ needs to be compared with the set $L_m(H)=L_{dmg}$ by considering all the possible $\sigma$-extensions $s'\sigma \in L(G)$ of $s'$ in $G$. 

To compute the attacker's observation sequence $\hat{P}_{o, A}^{V}(P_o(s))$, we need to annotate the supervisor's observation sequence $P_o(s) \in P_o(L(S \lVert G))$ with the sequence of control commands issued by the supervisor. Now, the basic idea is to first annotate all the string executions in $P_o(S)$ and then transfer the annotation to $S \lVert G$. With a dedicated product operation, one can obtain a transducer structure that maps string executions in $L(S \lVert G)$ to attacker's observation sequences, that is, a finite state structure encoding the $\hat{P}_{o, A}^{V}\circ P_o$ function. Since the attacker observes nothing when an unobservable event to the supervisor is executed, epsilon label can occur in the attacker's observation annotation\footnote{The attacker's observation annotation consists of $\epsilon$ label and those labels from the attacker's observation alphabet $(\Sigma_{o, A} \cup \{\epsilon\})\times \Gamma$; those events in $\Sigma_{uo}$ are mapped to the $\epsilon$ label, while the events in $\Sigma_{o}$ are mapped to the labels in $(\Sigma_{o, A} \cup \{\epsilon\})\times \Gamma$.} in the transducer structure for $\hat{P}_{o, A}^{V}\circ P_o$. We need to determinize the transducer structure with respect to the attacker's observation alphabet; it involves a subset construction w.r.t. the attacker's observation alphabet. After the determization operation, we can compute the set of all the strings in $L(S \lVert G)$ that give rise to the same attacker's observation sequence. However, this is not yet sufficient. To determine whether the attacker can establish a successful attack, we need to incorporate $H$ in the synchronization to obtain a refined transducer-like structure, before the determinization is carried out. The synchronization is performed up to those strings in $L(S \lVert G)$ and those strings in $L(G)$ that are one-step $\sigma$-extensions of strings in $L(S \lVert G)$, for each $\sigma \in \Sigma_{c, A}$. 

The synthesis algorithm for the supremal successful attackers consists of the following three steps:
\begin{enumerate}
\item compute an annotated supervisor
\item compute the generalized synchronous product of the plant, annotated supervisor and the damage automaton
\item perform a subset construction on the generalized synchronous product w.r.t. the attacker's observation alphabet
\end{enumerate}
For example, let us consider the plant $G'$ shown in the left of Fig.~\ref{fig:PS} and the supervisor $S'$ shown in the right of Fig.~\ref{fig:PS}. We have $\Sigma=\{a, a', b, c, d, d'\}$, $\Sigma_o=\{a, a', c, d, d'\}$, $\Sigma_c=\Sigma_o$, $\Sigma_{c, A}=\{d, d'\}$ and $\Sigma_{o, A}=\{c, d, d'\}$. The colored state, i.e., the state 4, in the plant is the unique bad state to avoid. It is not difficult to verify that $S'$ is indeed a supervisor on $G'$ w.r.t. $(\Sigma_c, \Sigma_o)$ and the state avoidance property has been enforced, that is, in the synchronous product of $G'$ and $S'$, every state $(q, x) \in Q \times X$, where $q$ is the state 4, can never be reached. We shall illustrate the three steps of the synthesis algorithm using this example of $G'$ and $S'$ in the rest of this section.

\begin{figure}[h]
\centering
\hspace*{-1mm}
\captionsetup{justification=centering}
\includegraphics[width=3.4in, height=1.6in]{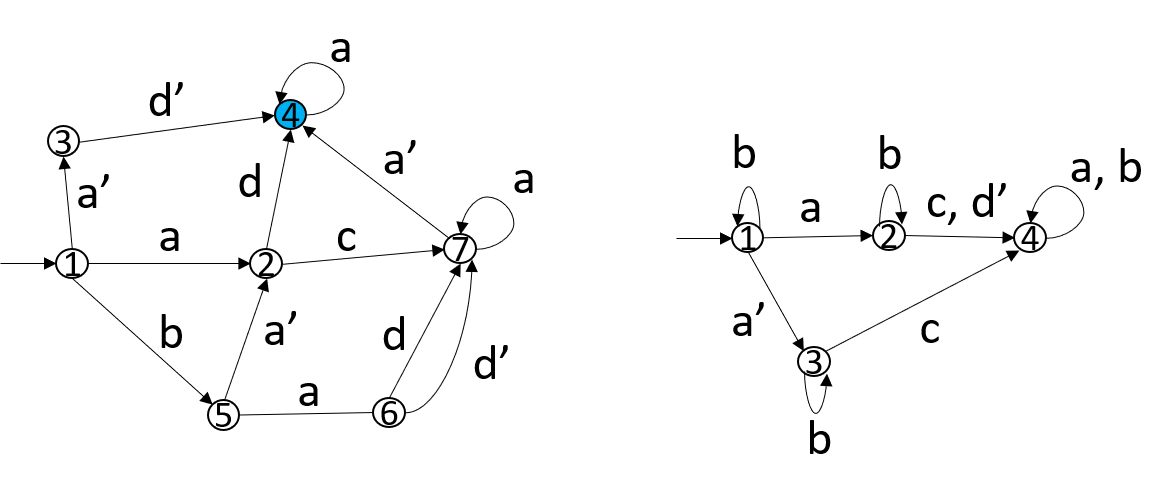} 
\caption{The plant $G'$ and the supervisor $S'$}
\label{fig:PS}
\end{figure}

\subsubsection{Annotation of the Supervisor}
Given the supervisor $S=(X, \Sigma, \zeta, x_0)$, we produce the annotated supervisor 
\begin{center}
$S^{A}=(X, \Sigma_o \times \Gamma \cup \Sigma_{uo}, \zeta^{A}, x_0)$, 
\end{center}
where $\zeta^{A}: X \times (\Sigma_o \times \Gamma \cup \Sigma_{uo}) \longrightarrow X$ is the partial transition function, which is defined as follows:
\begin{enumerate}
\item For any $x, x' \in X$, $\sigma \in \Sigma_o$, $\gamma\in \Gamma$, $\zeta_A(x, (\sigma, \gamma))=x'$ iff $\zeta(x, \sigma)=x'$ and $\gamma=\{\sigma \in \Sigma \mid \zeta(x', \sigma)!\}$.
\item For any $x, x' \in X$, $\sigma \in \Sigma_{uo}$, $\zeta_A(x, \sigma)=x'$ iff $\zeta(x, \sigma)=x'$.
\end{enumerate}
That is, we annotate each observable transition $(x, \sigma, x')$ in $\zeta$ with the control command $\gamma$ issued by the supervisor at state $x'$; that is, $(x, \sigma, x')$ is replaced by $(x, (\sigma, \gamma), x')$ in $\zeta^{A}$. Thus, if we project out the unobservable events in $\Sigma_{uo}$, each string in $S^{A}$ will be of the form $(\sigma_1, \gamma_1)(\sigma_2, \gamma_2)\ldots (\sigma_n, \gamma_n) \in (\Sigma_o \times \Gamma)^*$. From the construction of $S^{A}$, we have $V(\sigma_1\sigma_2 \ldots \sigma_i)=\{\sigma \in \Sigma \mid \zeta(\zeta(x_0, s'), \sigma)!, s' \in L(S) \cap P_o^{-1}(\sigma_1\sigma_2 \ldots \sigma_i)\}=\gamma_i$ for each $1\leq i \leq n$. Thus, if we project out the unobservable events in $\Sigma_{uo}$, each string in $S^{A}$ is of the form 
\begin{center}
$(\sigma_1, V(\sigma_1))(\sigma_2, V(\sigma_1\sigma_2))\ldots (\sigma_n, V(\sigma_1\sigma_2 \ldots \sigma_n)) \in (\Sigma_o \times \Gamma)^*$,
\end{center} where $\sigma_1\sigma_2\ldots \sigma_n$ is the supervisor's observation sequence. It is then straightforward to record all the attacker's observation sequences 
\begin{center}
$(P_{o, A}(\sigma_1), V(\sigma_1))(P_{o, A}(\sigma_2), V(\sigma_1\sigma_2))\ldots (P_{o, A}(\sigma_n), V(\sigma_1\sigma_2 \ldots \sigma_n)) \in ((\Sigma_{o, A} \cup \{\epsilon\})\times \Gamma)^*$
\end{center}
in $S^A$ afterwards, which will be delayed until the annotation in $P_o(S^A)$ is transferred to $S \lVert G$ via a dedicated synchronous product operation. 

For example, the annotated supervisor $S'^A$ of $S'$ in Fig.~\ref{fig:PS} is shown in Fig.~\ref{fig:Asupervisor}. If the plant generates $s=ba'c$, then $P_o(s)=a'c$ is the supervisor's observation sequence. It is clear that the attacker's observation sequence is $(\epsilon, \{b, c\})(c, \{b, a\})$. However, we only record the sequence $b(a', \{b, c\})(c, \{b, a\})$ in $S'^A$ or, equivalently, we record the sequence $(a', \{b, c\})(c, \{b, a\})$ in $P(S'^A)$.
\begin{figure}[h]
\centering
\hspace*{-1mm}
\captionsetup{justification=centering}
\includegraphics[width=3.0in, height=1.2in]{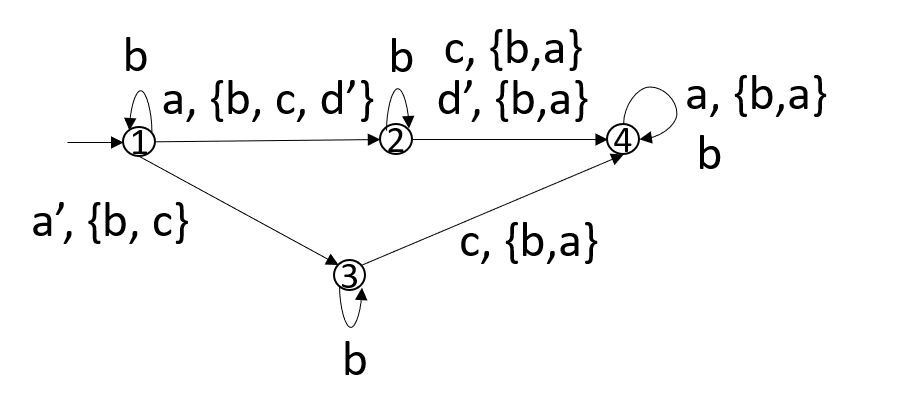} 
\caption{The annotated supervisor $S'^A$}
\label{fig:Asupervisor}
\end{figure}

We need to transfer the annotation in $P(S^A)$ to $S \lVert G$ to obtain a transducer structure that maps string executions in $L(S \lVert G)$ to attacker's observation sequences. This is not that difficult as we only need to synchronize the plant $G$ and the annotated supervisor $S^A$ using a dedicated synchronous product operation, which transfers the attacker's observation annotation, which then encodes the function $\hat{P}_{o, A}^{V} \circ P_o: L(S \lVert G) \rightarrow ((\Sigma_{o, A} \cup \{\epsilon\}) \times \Gamma)^*$. Based on the product of $G$ and $S^A$, we can compute every set of strings in $L(S \lVert G)$ that can be mapped to the same attacker's observation sequence, via a subset construction w.r.t. the attacker's observation alphabet. To determine whether an attack shall be established for an attacker's observation sequence, we need to synchronize $G$ and $S^A$ with $H$ before the determinization is performed. The overall product operation will be referred to as the {\em generalized synchronous product} operation, which is explained in detail in the next subsection.

\subsubsection{Generalized Synchronous Product}
Given the plant $G=(Q, \Sigma, \delta, q_0)$, the annotated supervisor $S^{A}=(X, \Sigma_o \times \Gamma \cup \Sigma_{uo}, \zeta^{A}, x_0)$ and the damage automaton $H=(Z, \Sigma, \eta, z_0, Z_m)$, the generalized synchronous product $GP(G, S^A, H)$ of $G$, $S^{A}$ and $H$ is given by
\begin{center}
$(Q \times X \times Z \cup \{\bot, \top\}, \Sigma^{GP}, \delta^{GP}, (q_0, x_0, z_0))$,
\end{center}
where $\bot, \top$ are two new states that are different from the states in $Q \times X \times Z$, $\Sigma^{GP}=\Sigma_o \times ((\Sigma_{o, A} \cup \{\epsilon\}) \times \Gamma) \cup \Sigma_{uo} \times \{\epsilon\} \cup \Sigma_{c, A}$ and the partial transition function 
\begin{center}
$\delta^{GP}: (Q \times X \times Z \cup \{\bot, \top\}) \times \Sigma^{GP} \longrightarrow Q \times X \times Z \cup \{\bot, \top\}$
\end{center}
 is defined as follows: 
\begin{enumerate}
\item For any $q, q' \in Q$, $x, x' \in X$, $z, z' \in Z$, $\gamma \in \Gamma$, $\sigma \in \Sigma_{o, A}$, $\delta^{GP}((q, x, z), (\sigma, (\sigma, \gamma)))=(q', x', z')$ iff $\delta(q, \sigma)=q'$, $\zeta^A(x, (\sigma, \gamma))=x'$ and $\eta(z, \sigma)=z'$.
\item For any $q, q' \in Q$, $x, x' \in X$, $z, z' \in Z$, $\gamma \in \Gamma$, $\sigma \in \Sigma_{o}-\Sigma_{o, A}$, $\delta^{GP}((q, x, z), (\sigma, (\epsilon, \gamma)))=(q', x', z')$ iff $\delta(q, \sigma)=q'$, $\zeta^A(x, (\sigma, \gamma))=x'$ and $\eta(z, \sigma)=z'$.
\item For any $q, q' \in Q$, $x, x' \in X$, $z, z' \in Z$, $\sigma \in \Sigma_{uo}$, $\delta^{GP}((q, x, z), (\sigma, \epsilon))=(q', x', z')$ iff $\delta(q, \sigma)=q'$, $\zeta^A(x, \sigma)=x'$ and $\eta(z, \sigma)=z'$.
\item For any $q \in Q$, $x \in X$, $z\in Z$, $\sigma \in \Sigma_{c, A}$, $\delta^{GP}((q, x, z), \sigma)=\top$ iff $\delta(q, \sigma)!$, $\neg \zeta(x, \sigma)!$ and $\eta(z, \sigma) \in Z_m$.
\item For any $q \in Q$, $x \in X$, $z\in Z$, $\sigma \in \Sigma_{c, A}$, $\delta^{GP}((q, x, z), \sigma)=\bot$ iff $\delta(q, \sigma)!$, $\neg \zeta(x, \sigma)!$ and $\eta(z, \sigma) \notin Z_m$.
\end{enumerate}
The definition of $GP(G, S^A, H)$ looks complicated. We shall analyze each item carefully. The state space is $Q \times X \times Z \cup \{\bot, \top\}$. Intuitively, the state $\bot$ indicates a failed attack, while the state $\top$ indicates a successful attack. The alphabet $\Sigma^{GP}$ is the union of 
\begin{center}
$\Sigma_o \times ((\Sigma_{o, A} \cup \{\epsilon\}) \times \Gamma)$, $\Sigma_{uo} \times \{\epsilon\}$ and $\Sigma_{c, A}$.
\end{center}
 For each element in $\Sigma_o \times ((\Sigma_{o, A} \cup \{\epsilon\}) \times \Gamma)$ and $\Sigma_{uo} \times \{\epsilon\}$, the first component $\sigma \in \Sigma$ is used to denote an event executed by the plant and the second component $l$ is used to denote the corresponding observation (annotation) by the attacker. In particular, if $\sigma \in \Sigma_{o, A} \subseteq \Sigma_o$ is executed by the plant, then the attacker's observation is $l=(\sigma, \gamma) \in \Sigma_{o, A} \times \Gamma$ for some $\gamma \in \Gamma$; the attacker's observation is $l=(\epsilon, \gamma) \in \{\epsilon\} \times \Gamma$ for some $\gamma \in \Gamma$ if $\sigma \in \Sigma_o-\Sigma_{o, A}$ is executed by the plant. If $\sigma \in \Sigma_{uo}$ is executed by the plant, then the attacker observes nothing and we use the $l=\epsilon$ label for the annotation. 

By the items 1)-3) of the definition of $\delta^{GP}$, we can observe that every string $s \in L(S \lVert G)$ is the projection of some string of $L(GP(G, S^A, H))$ onto the first component. Moreover, for any string of the form $\overline{s}=(\sigma_1, l_1)(\sigma_2, l_2)\ldots (\sigma_n, l_n)$ in $L(GP(G, S^A, H))$, we 
have that $l_1l_2 \ldots l_n=\hat{P}_{o, A}^{V}(P_o(\sigma_1\sigma_2\ldots \sigma_n))$. That is, the items 1)-3) of the definition of $\delta^{GP}$ ensures that part of $GP(G, S^A, H)$ encodes the transducer structure that maps string executions in $L(S \lVert G)$ to attacker's observation sequences. By the items 4)-5) of the definition of $\delta^{GP}$, we observe that, when projected onto the first component, $L(GP(G, S^A, H))$ can generate strings in $L(G)-L(S \lVert G)$, but these strings can only be some one-step $\sigma$-extensions of strings in $L(S \lVert G)$ for $\sigma \in \Sigma_{c, A}$; in particular, these strings occur precisely because of the actuator attack on attackable events. After the actuator attack, the closed-loop system will be halted and thus we will not need to record the attacker's observation any more. It is for this reason that $\Sigma_{c, A}$ is included in $\Sigma^{GP}$. $GP(G, S^A, H)$ transduces strings in $L(S \lVert G)$ and, when projected onto the first component, record the strings in $L(S \Vert G)\Sigma_{c, A} \cap L(G)-L(S \lVert G)$. Thus, for each string $\overline{s}$ of the form $(\sigma_1, l_1)(\sigma_2, l_2)\ldots (\sigma_n, l_n)\sigma_{n+1}$ generated by $GP(G, S^A, H)$, where $\sigma_{n+1} \in  \Sigma_{c, A}$, we have that $\sigma_1\sigma_2\ldots \sigma_{n+1} \in L_{dmg}$ if $\delta^{GP}((q_0, x_0, z_0), \overline{s})=\top$ and $\sigma_1\sigma_2\ldots \sigma_{n+1} \notin L_{dmg}$ if $\delta^{GP}((q_0, x_0, z_0), \overline{s})=\bot$; in both cases, we have $\sigma_1\sigma_2\ldots \sigma_{n+1} \in L(S \Vert G)\Sigma_{c, A} \cap L(G)-L(S \lVert G)$. 


For example, consider the damage automaton $H'$ shown in Fig.~\ref{fig:damage}, where we have omitted the dump state, i.e., state 6, and the corresponding transitions. 
\begin{figure}[h]
\centering
\hspace*{-1mm}
\captionsetup{justification=centering}
\includegraphics[width=2.2in, height=1.0in]{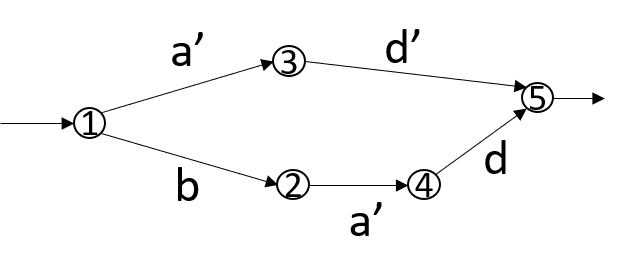} 
\caption{The damage automaton $H'$, where the dump state is not explicitly shown}
\label{fig:damage}
\end{figure}
The generalized synchronous product $GP(G', S'^A, H')$ of $G'$, $S'^{A}$ and $H'$ is then shown in Fig.~\ref{fig:GSP}. The set of damaging strings is $L'_{dmg}=\{a'd', ba'd\}$. 

\begin{figure}[h]
\centering
\hspace*{-1mm}
\captionsetup{justification=centering}
\includegraphics[width=5.2in, height=3.8in]{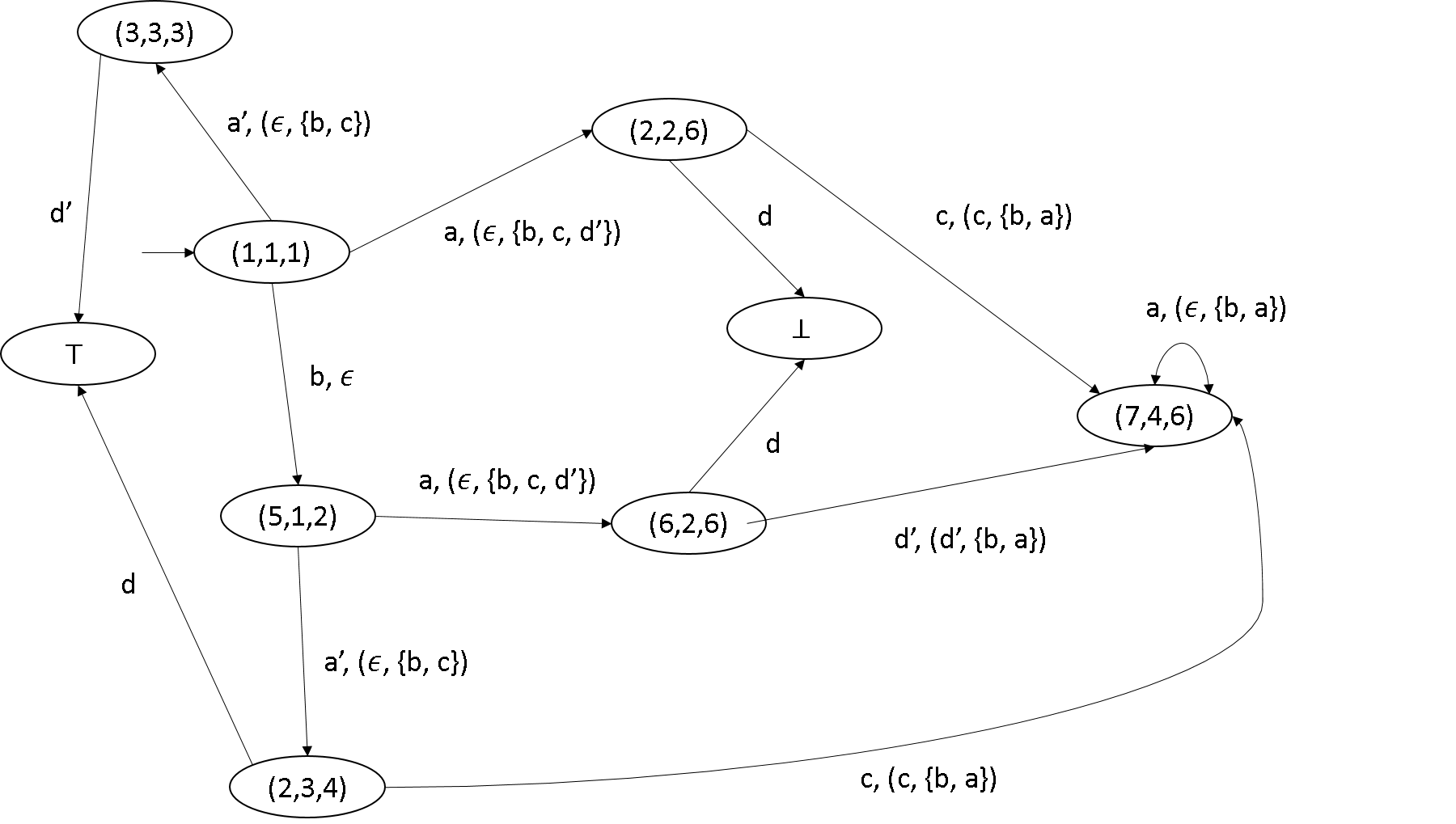} 
\caption{The generalized synchronous product $GP(G', S'^A, H')$}
\label{fig:GSP}
\end{figure}

\subsubsection{Subset Construction w.r.t. the Attacker's Observation Alphabet}
The generalized synchronous product $GP(G, S^A, H)$ is a transducer-like structure that maps string executions in $L(S \lVert G)$ to attacker's observation sequences; furthermore, it records in all the cases when an actuator attack occurs and also determines whether it will succeed. To determine whether a successful actuator attack can be established, we need to determine each set\footnote{By incorporating $H$ in the generalized synchronous product, we do not even need to explicitly construct this set: we can directly compare this set of strings with the set $L_m(H)=L_{dmg}$ by considering all the possible $\sigma$-extensions of these strings in $G$ with subset construction.} of string executions in $L(S \lVert G)$ that are mapped to the same attackers' observation sequence. This is what will be carried out in this subsection, which involves a subset construction w.r.t. the attacker's observation alphabet. 

{\bf Subset Construction w.r.t. the Attacker's Observation Alphabet}: 

Given $GP(G, S^A, H)$, we shall now perform subset construction on $GP(G, S^A, H)$ with respect to the attacker's observation alphabet. 

Let $GPS(G, S^A, H)$ denote the sub-automaton of $GP(G, S^A, H)$ with state space restricted to $Q \times X \times Z$. The alphabet of $GPS(G, S^A, H)$ is $\Sigma^{GPS}=\Sigma_o \times ((\Sigma_{o, A} \cup \{\epsilon\}) \times \Gamma) \cup \Sigma_{uo} \times \{\epsilon\}$ and it is clear that the partial transition function $\delta^{GPS}: (Q \times X \times Z) \times \Sigma^{GPS} \mapsto Q \times X \times Z$ of $GPS(G, S^A, H)$ is defined as follows: 
\begin{enumerate}
\item For any $q, q' \in Q$, $x, x' \in X$, $z, z' \in Z$, $\gamma \in \Gamma$, $\sigma \in \Sigma_{o, A}$, $\delta^{GPS}((q, x, z), (\sigma, (\sigma, \gamma)))=(q', x', z')$ iff $\delta(q, \sigma)=q'$, $\zeta^A(x, (\sigma, \gamma))=x'$ and $\eta(z, \sigma)=z'$.
\item For any $q, q' \in Q$, $x, x' \in X$, $z, z' \in Z$, $\gamma \in \Gamma$, $\sigma \in \Sigma_{o}-\Sigma_{o, A}$, $\delta^{GPS}((q, x, z), (\sigma, (\epsilon, \gamma)))=(q', x', z')$ iff $\delta(q, \sigma)=q'$, $\zeta^A(x, (\sigma, \gamma))=x'$ and $\eta(z, \sigma)=z'$.
\item For any $q, q' \in Q$, $x, x' \in X$, $z, z' \in Z$, $\sigma \in \Sigma_{uo}$, $\delta^{GPS}((q, x, z), (\sigma, \epsilon))=(q', x', z')$ iff $\delta(q, \sigma)=q'$, $\zeta^A(x, \sigma)=x'$ and $\eta(z, \sigma)=z'$.
\end{enumerate}
Let $GPS^1(G, S^A, H)$ (respectively, $GPS^2(G, S^A, H)$) denote the automaton that is obtained from $GPS(G, S^A, H)$ by removing\footnote{In general, $GPS^2(G, S^A, H)$ is an $\epsilon$-nondeterministic finite automaton~\cite{HU79}.} the second component $l$ (respectively, first component $\sigma$) from the event $(\sigma, l)$ for each transition. Let $SUB(GPS^2(G, S^A, H))$ denote the automaton that is obtained from $GPS^2(G, S^A, H)$ by subset construction~\cite{HU79}. It follows that $SUB(GPS^2(G, S^A, H))=(Y, \Sigma^{SUB}, \Delta^{SUB}, y_0)$, where $Y=2^{Q \times X \times Z}$, $y_0=\delta^{GPS}((q_0, x_0, z_0), (\Sigma_{uo} \times \{\epsilon\})^*) \in Y$, $\Sigma^{SUB}=(\Sigma_{o, A} \cup \{\epsilon\}) \times \Gamma$, $\Delta^{SUB}: Y \times \Sigma^{SUB} \longrightarrow Y$ is the transition function defined as follows: for any $y \in Y$, $l \in \Sigma^{SUB}$, $\Delta^{SUB}(y, l):=$
\begin{center}
$\delta^{GPS}(\delta^{GPS}(y, \Sigma_o \times \{l\}), (\Sigma_{uo} \times \{\epsilon\})^*)$
\end{center}
Let $SUB(GP(G, S^A, H)):=(SUB(GPS^2(G, S^A, H)), L_f)$, where $Lf: Y \mapsto 2^{\Sigma_{c, A}}$ is a labeling function defined as follows: for any $y \in Y$ and any $\sigma \in \Sigma_{c, A}$, $\sigma \in Lf(y)$ iff there exists some $v \in y$ such that,
\begin{enumerate}
\item $\delta^{GP}(v, \sigma)=\top$, and
\item for any $v' \in y$, $\delta^{GP}(v', \sigma)!$ implies $\delta^{GP}(v', \sigma)=\top$.
\end{enumerate}

Intuitively, $SUB(GP(G, S^A, H))$ is a Moore automaton that maps each attacker's observation sequence to the subset of attackable events that are attacked\footnote{An attackable event is attacked by $A^{sup}$ if it is disabled by the supervisor and is enabled by $A^{sup}$.} by $A^{sup}$. In particular, after the attacker observes the sequence $l_1l_2\ldots l_n \in ((\Sigma_{o, A} \cup \{\epsilon\}) \times \Gamma)^*=(\Sigma^{SUB})^*$, the subset of attackable events that are attacked by $A^{sup}$ is $Lf(\Delta^{SUB}(y_0, l_1l_2\ldots l_n))$. If $Lf(\Delta^{SUB}(y_0, l_1l_2\ldots l_n))=\varnothing$, then the attacker cannot establish an attack after observing $l_1l_2\ldots l_n$. If this holds for all the possible attacker's observation sequences in $\hat{P}_{o, A}^{V}(P_o(L(S \lVert G)))$, then we conclude that $(G, S)$ is not attackable w.r.t. $(\Sigma_{c, A}, \Sigma_{o, A})$ and $L_{dmg}$.

\subsubsection{Proof of Correctness}

In this subsection, we shall show that the supremal successful attacker $A^{sup}: \hat{P}_{o, A}^{V}(P_o(L(S \lVert G))) \mapsto \Delta \Gamma$ is realized\footnote{A formal interpretation will be provided soon.} by the Moore automaton $SUB(GP(G, S^A, H))$, if $(G, S)$ is attackable w.r.t. $(\Sigma_{c, A}, \Sigma_{o, A})$ and $L_{dmg}$. If $(G, S)$ is not attackable w.r.t. $(\Sigma_{c, A}, \Sigma_{o, A})$ and $L_{dmg}$, then we can also directly read this information off $SUB(GP(G, S^A, H))$. 


The next lemma shows that the domain of $A^{sup}$ is indeed captured by $SUB(GP(G, S^A, H))$. The closed-behavior $L(SUB(GP(G, S^A, H)))$ of the Moore automaton $SUB(GP(G, S^A, H))=(SUB(GPS^2(G, S^A, H)), L_f)$ is defined to be $L(SUB(GPS^2(G, S^A, H)))$\footnote{That is, the labeling function is ignored in the definition of the closed-behavior of an Moore automaton.}. 
\begin{lemma}
$L(SUB(GP(G, S^A, H)))=\hat{P}_{o, A}^{V}(P_o(L(S \lVert G)))$.
\end{lemma}
\begin{proof}
It is clear that $L(SUB(GPS^2(G, S^A, H)))=L(GPS^2(G, S^A, H))$\footnote{The definition of the closed-behavior of an $\epsilon$-nondeterministic finite automaton ($\epsilon$-NFA) is similar to that of a deterministic finite automaton~\cite{HU79}.}, since subset construction preserves the closed-behavior of an $\epsilon$-NFA. Recall that $GPS^2(G, S^A, H)$ denotes the $\epsilon$-NFA that is obtained from $GPS(G, S^A, H)$ by removing the first component $\sigma$ from the event $(\sigma, l)$ for each transition. By the construction of $GPS(G, S^A, H)$, $L(GPS(G, S^A, H))$ consists of strings of the form $\overline{s}=(\sigma_1, l_1)(\sigma_2, l_2)\ldots (\sigma_n, l_n)$, where $l_1l_2 \ldots l_n=\hat{P}_{o, A}^{V}(P_o(\sigma_1\sigma_2\ldots \sigma_n))$ and $\sigma_1\sigma_2\ldots \sigma_n \in L(S \lVert G)$; moreover, every string $s$ in $L(S \lVert G)$ is the projection of some string in $L(GPS(G, S^A, H))$ onto the first component, by the construction of $GPS(G, S^A, H)$. Thus, we have $L(GPS^1(G, S^A, H))=L(S \lVert G)$ and $L(GPS^2(G, S^A, H))=\hat{P}_{o, A}^{V}(P_o(L(S \lVert G)))$. Thus, $L(SUB(GP(G, S^A, H)))=\hat{P}_{o, A}^{V}(P_o(L(S \lVert G)))$.
\end{proof}
\begin{proposition}
\label{propo: existence}
$AP_{dmg} \neq \varnothing$ iff there exists a reachable state $y$ of $SUB(GP(G, S^A, H))$ with $Lf(y) \neq \varnothing$.
\end{proposition}
\begin{proof}
Suppose there exists a reachable state $y$ in $SUB(GP(G, S^A, H))$ with $Lf(y) \neq \varnothing$. Let $\sigma \in Lf(y)$. By the definition, $\sigma \in \Sigma_{c, A}$ and there exists some $v \in y$ such that, 1) $\delta^{GP}(v, \sigma)=\top$ and 2) for any $v' \in y$, $\delta^{GP}(v', \sigma)!$ implies $\delta^{GP}(v', \sigma)=\top$. Then, by the construction of $GPS^1(G, S^A, H)$ and $SUB(GPS^2(G, S^A, H))$, there exists some string $s \in L(S \lVert G)$ such that 
\begin{enumerate}
\item $\delta^{GPS1}((q_0, x_0, z_0), s)=v$, and 
\item for any $v' \in y$, there exists some string $s' \in L(S \lVert G)$ such that $\delta^{GPS1}((q_0, x_0, z_0), s')=v'$ and $\hat{P}_{o, A}^{V}(P_o(s))=\hat{P}_{o, A}^{V}(P_o(s'))$; moreover, for any $s' \in L(S \lVert G)$ such that $\hat{P}_{o, A}^{V}(P_o(s))=\hat{P}_{o, A}^{V}(P_o(s'))$, we have $\delta^{GPS1}((q_0, x_0, z_0), s') \in y$. 
\end{enumerate}
Here, $\delta^{GPS1}$ denotes the partial transition function of $GPS^1(G, S^A, H)$. We shall now show that $(s, \sigma)$ is an attack pair for $(G, S)$ w.r.t. $(\Sigma_{c, A}, \Sigma_{o, A})$ and $L_{dmg}$. It then follows that $AP_{dmg} \neq \varnothing$. 

Indeed, $s\sigma \in L_{dmg}$ since $\delta^{GPS1}((q_0, x_0, z_0), s)=v$ and $\delta^{GP}(v, \sigma)=\top$. For any $s' \in L(S \lVert G)$ such that $\hat{P}_{o, A}^{V}(P_o(s))=\hat{P}_{o, A}^{V}(P_o(s'))$ and $s'\sigma \in L(G)$, we have $\delta^{GPS1}((q_0, x_0, z_0), s') \in y$, $s'\sigma \in L(G)$ and $s'\sigma \notin L(S \lVert G)$, since $s\sigma \notin L(S \lVert G)$ and $V(P_o(s))=V(P_o(s'))$. Thus, we conclude that $\delta^{GP}(\delta^{GPS1}((q_0, x_0, z_0), s'), \sigma)!$ and thus $\delta^{GP}(\delta^{GPS1}((q_0, x_0, z_0), s'), \sigma)=\top$. We then conclude that $s'\sigma \in L_{dmg}$ and then $(s, \sigma)$ is an attack pair for $(G, S)$ w.r.t. $(\Sigma_{c, A}, \Sigma_{o, A})$ and $L_{dmg}$.


Suppose there exists an attack pair $(s, \sigma) \in AP_{dmg}$. Then, by definition, $s\sigma \in L_{dmg}$ and, for any $s' \in L(S \lVert G)$, $\hat{P}_{o, A}^{V}(P_o(s))=\hat{P}_{o, A}^{V}(P_o(s'))$ and $s'\sigma \in L(G)$ together implies $s'\sigma \in L_{dmg}$. It is clear that $\Delta^{SUB}(y_0, \hat{P}_{o, A}^{V}(P_o(s)))!$ since $\delta^{GPS1}((q_0, x_0, z_0), s)!$. We now show that $\sigma \in Lf(\Delta^{SUB}(y_0, \hat{P}_{o, A}^{V}(P_o(s))))$. It then follows that $Lf(\Delta^{SUB}(y_0, \hat{P}_{o, A}^{V}(P_o(s)))) \neq \varnothing$ and $\Delta^{SUB}(y_0, \hat{P}_{o, A}^{V}(P_o(s)))$ is a reachable state of $SUB(GP(G, S^A, H))$.

First of all, we have $\sigma \in \Sigma_{c, A}$ by the definition of an attack pair. It is clear that $\delta^{GPS1}((q_0, x_0, z_0), s) \in \Delta^{SUB}(y_0, \hat{P}_{o, A}^{V}(P_o(s)))$. We only need to show the following:
\begin{enumerate}
\item $\delta^{GP}(\delta^{GPS1}((q_0, x_0, z_0), s), \sigma)=\top$, and
\item for any $v' \in \Delta^{SUB}(y_0, \hat{P}_{o, A}^{V}(P_o(s)))$, $\delta^{GP}(v', \sigma)!$ implies $\delta^{GP}(v', \sigma)=\top$. 
\end{enumerate}
Then, it follows that $\sigma \in Lf(\Delta^{SUB}(y_0, \hat{P}_{o, A}^{V}(P_o(s))))$. 

The first item is straightforward from $s\sigma \in L_{dmg}$. In the rest, we prove the second item.

Let $v' \in \Delta^{SUB}(y_0, \hat{P}_{o, A}^{V}(P_o(s)))$. It follows that there exists some string $s' \in L(S \lVert G)$ such that $\delta^{GPS1}((q_0, x_0, z_0), s')=v'$ and $\hat{P}_{o, A}^{V}(P_o(s))=\hat{P}_{o, A}^{V}(P_o(s'))$. Now, we suppose $\delta^{GP}(v', \sigma)!$, then we have $s'\sigma \in L(G)$. That is, we have $s' \in L(S \lVert G)$, $\hat{P}_{o, A}^{V}(P_o(s))=\hat{P}_{o, A}^{V}(P_o(s'))$ and $s'\sigma \in L(G)$. Thus, we can then conclude that $s'\sigma \in L_{dmg}$. That is, we have $\delta^{GP}(\delta^{GPS1}((q_0, x_0, z_0), s'), \sigma)=\top$, i.e., $\delta^{GP}(v', \sigma)=\top$.

\end{proof}

Based on Proposition~\ref{propo: existence}, we immediately have the following theorem. 
\begin{theorem}
$(G, S)$ is attackable with respect to $(\Sigma_{c, A}, \Sigma_{o, A})$ and $L_{dmg}$ iff there exists a reachable state $y$ in $SUB(GP(G, S^A, H))$ with $Lf(y) \neq \varnothing$.
\end{theorem}
Furthermore, if $(G, S)$ is attackable w.r.t. $(\Sigma_{c, A}, \Sigma_{o, A})$ and $L_{dmg}$, then $A^{sup}:\hat{P}_{o, A}^{V}(P_o(L(S \lVert G))) \mapsto \Delta \Gamma$ is realized by $SUB(GP(G, S^A, H))$ in the following sense.
\begin{theorem}
Suppose $AP_{dmg} \neq \varnothing$. Then, $I(\hat{P}_{o, A}^{V}(P_o(s)))= Lf(\Delta^{SUB}(y_0, \hat{P}_{o, A}^{V}(P_o(s))))$ for any $s \in L(S \lVert G)$.
\end{theorem}
\begin{remark}
Recall that $A^{sup}(\hat{P}_{o, A}^{V}(P_o(s)))=(V(P_o(s)) \cap \Sigma_{c, A}) \cup I(\hat{P}_{o, A}^{V}(P_o(s)))$, for any $s \in L(S \lVert G)$, where $I(\hat{P}_{o, A}^{V}(P_o(s))):=\{\sigma \in \Sigma_{c, A} \mid \exists s' \in L(S \lVert G), \hat{P}_{o, A}^{V}(P_o(s))=\hat{P}_{o, A}^{V}(P_o(s')) \wedge (s', \sigma) \in AP_{dmg}\}$. It is not difficult to see that $I(\hat{P}_{o, A}^{V}(P_o(s)))$ is indeed the set of attackable events that are attacked by $A^{sup}$ after observing the attacker's observation sequence $\hat{P}_{o, A}^{V}(P_o(s))$. Thus, this theorem states that, after the sequence $\hat{P}_{o, A}^{V}(P_o(s))$ is observed, the subset of attackable events that are attacked by $A^{sup}$ is $Lf(\Delta^{SUB}(y_0, \hat{P}_{o, A}^{V}(P_o(s))))$.
\end{remark}
\begin{proof}
Let $\sigma \in I(\hat{P}_{o, A}^{V}(P_o(s)))$. Then, by the definition, there exists some $s' \in L(S \lVert G)$ such that $\hat{P}_{o, A}^{V}(P_o(s))=\hat{P}_{o, A}^{V}(P_o(s'))$ and $(s', \sigma) \in AP_{dmg}$. From the proof of Proposition~\ref{propo: existence}, we conclude that $\sigma \in Lf(\Delta^{SUB}(y_0, \hat{P}_{o, A}^{V}(P_o(s'))))=Lf(\Delta^{SUB}(y_0, \hat{P}_{o, A}^{V}(P_o(s))))$. 

Let $\sigma \in Lf(\Delta^{SUB}(y_0, \hat{P}_{o, A}^{V}(P_o(s))))$. By the definition, $\sigma \in \Sigma_{c, A}$ and there exists some $s' \in L(S \lVert G)$ such that 
\begin{enumerate}
\item $\hat{P}_{o, A}^{V}(P_o(s'))))=\hat{P}_{o, A}^{V}(P_o(s))))$, $\delta^{GPS1}((q_0, x_0, z_0), s') \in \Delta^{SUB}(y_0, \hat{P}_{o, A}^{V}(P_o(s)))$ and $\delta^{GP}(\delta^{GPS1}((q_0, x_0, z_0), s'), \sigma)=\top$
\item for any $s'' \in L(S \lVert G)$ such that $\hat{P}_{o, A}^{V}(P_o(s''))))=\hat{P}_{o, A}^{V}(P_o(s))))$, we have $\delta^{GPS1}((q_0, x_0, z_0), s'') \in \Delta^{SUB}(y_0, \hat{P}_{o, A}^{V}(P_o(s)))$,
and $\delta^{GP}(\delta^{GPS1}((q_0, x_0, z_0), s''),\sigma)!$ implies $\delta^{GP}(\delta^{GPS1}((q_0, x_0, z_0), s''),\sigma)=\top$
\end{enumerate}
The first item says that $s'\sigma \in L_{dmg}$ and the second item says that for any $s'' \in L(S \lVert G)$ such that $\hat{P}_{o, A}^{V}(P_o(s''))))=\hat{P}_{o, A}^{V}(P_o(s'))))$, $s''\sigma \in L(G)$ implies $s''\sigma \in L_{dmg}$. Thus, we conclude that $(s', \sigma)$ is an attack pair w.r.t. $(\Sigma_{c, A}, \Sigma_{o, A})$ and $L_{dmg}$. Since $s' \in L(S \lVert G)$ and $\hat{P}_{o, A}^{V}(P_o(s'))))=\hat{P}_{o, A}^{V}(P_o(s))))$, we conclude that $\sigma \in I(\hat{P}_{o, A}^{V}(P_o(s)))$.
\end{proof}

For example, the Moore automaton $SUB(GP(G', S'^A, H'))$ is shown in Fig.~\ref{fig:final} for the $G'$, $S'$ and $H'$ given previously. In particular, we can conclude that $(G', S')$ is attackable with respect to 
$(\Sigma_{c, A}, \Sigma_{o, A})$ and $L'_{dmg}$, since 
\begin{center}
$Lf(\{(3,3,3), (2,3,4)\}) \neq \varnothing$.
\end{center} 
The supremal successful attacker $A'^{sup}$ performs attack after observing $\overline{l}=(\epsilon, \{b, c\})$ (the attacker can conclude that string $s=a'$ or $s=ba'$ has been generated), and the attacker enables attackable events $d, d'$ and leads to successful attack ($I(\epsilon, \{b, c\})=\{d, d'\}$).

\begin{figure}[h]
\centering
\hspace*{-1mm}
\captionsetup{justification=centering}
\includegraphics[width=4.0in, height=2.4in]{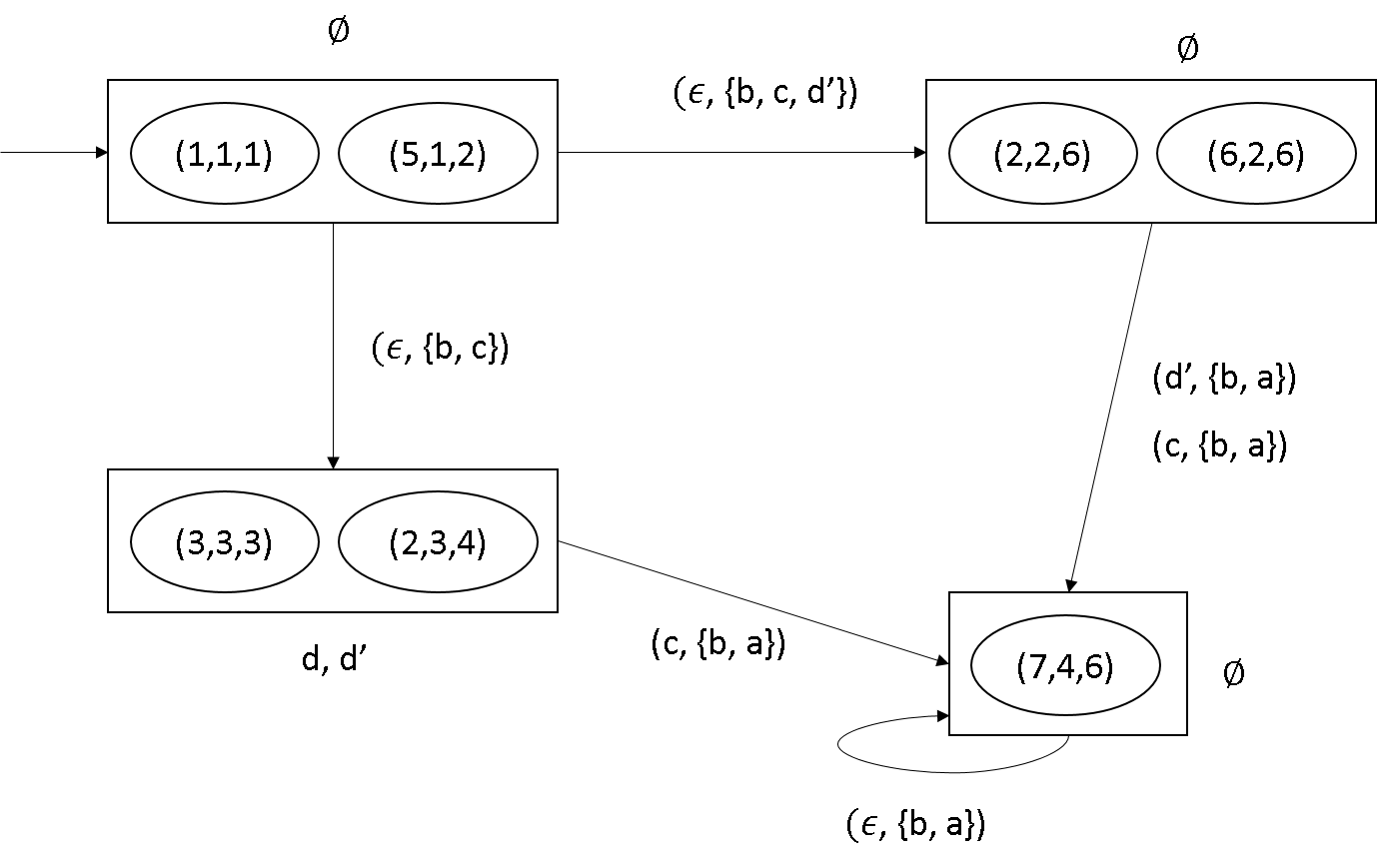} 
\caption{A Moore automaton representation of the supremal successful enabling attacker $A^{sup}$}
\label{fig:final}
\end{figure}

\section{Discussions and Conclusions}
\label{section:DC}
In this work, we have formulated the following two synthesis problems within discrete-event systems formalism.
\begin{enumerate}
\item the problem of synthesis of successful actuator attackers on supervisors. 
\item the problem of synthesis of resilient supervisors against actuator attackers. 
\end{enumerate}
Moreover, we have resolved the problem of synthesis of successful actuator attackers on normal supervisors under the assumption $\Sigma_{c, A} \subseteq \Sigma_{o, A}$, i.e., a normality assumption on the attacker, in which case the supremal successful actuator attackers exist and can be synthesized based on the synthesis algorithm provided in Section~\ref{subsection:SSSEA}. 

This research work imposes some reasonable assumptions (see Remark~\ref{remark: asspt}) in order for the exposition to remain elementary. In the following, we shall briefly discuss about some future research directions that could be continued based on this work.

First of all, the supremal actuator attackers synthesized in this work are considered to be passive. They are passive since they would never disable attackable events that are enabled by the supervisors (i.e., they are enabling). A passive attacker never alters the execution of the closed-loop system within $L(V/G)$, and it only patiently waits for its chance in order to establish a successful attack; on the other hand, an active attacker will influence the execution of the closed-loop system within $L(V/G)$, by properly disabling attackable events that are enabled by the supervisors, so that there is a higher chance for it to establish a successful attack in any single run. From this point of view, the attack goal 
\begin{enumerate}
\item $L(V_A/G) \cap L_{dmg}\neq \varnothing$
\item $L(V_A/G)-P_{o}^{-1}P_o(L(V/G))\subseteq L_{dmg}\Sigma^*$.
\end{enumerate}
is to some extent actually a {\em weak attack goal}, and then the attackability formulated in this paper corresponds to {\em weak attackability}. It is of interest to impose other conditions so that the synthesized supervisors can achieve {\em strong attack goals} and we need to formulate a notion of {\em strong attackability} correspondingly\footnote{Alternatively, we may refer to weak attack as passive attack and strong attack as active attack.}. Intuitively, in strong attack, the attacker may disable attackable events that are enabled by the supervisors for two different purposes: 1) increase the chance that the attack will be successful, 2) increase the chances that the closed-loop system will execute a vulnerable string\footnote{a string $s \in L(V/G)$ is said to be {\em vulnerable} if there exists some $\sigma \in \Sigma_{c, A}$ such that $(s, \sigma) \in AP_{dmg}$.} that can be attacked. 

Secondly, we assume that the attackers will not take risks in attacking the supervisors. This greatly limits the attack capability of the attackers. In many scenarios, it is useful to consider the synthesis of risky attackers. An example of a risky attack is to attack an attackable event when the attacker knows it is possible but is unsure whether a damaging string will be generated or not.  

Thirdly, we impose in this work the assumption $\Sigma_{c, A} \subseteq \Sigma_{o, A}$ and only consider the case of attacking normal supervisors. It is of interest to relax these two assumptions and study under what circumstances can supremal successful attackers be synthesized in general. 

Finally, the main purpose for studying the synthesis of successful attackers is to synthesize resilient supervisor against actuator attackers. Thus, it is crucial to address the problem of synthesis of resilient supervisors, which is also left as a future work. 

It is also of interest to relax assumption that the set $\Sigma_{o, A}$ of events observable to the attacker is a subset of the set $\Sigma_o$ of observable events for the supervisor. It is also of interest to consider attack architectures with distributed or hierarchical architectures~\cite{WMW10}. 

\begin{acknowledgements}
This work is financially supported by Singapore Ministry of Education Academic Research Grant RG91/18-(S)-SU RONG (VP), which is gratefully acknowledged. 
\end{acknowledgements}

\addtolength{\textheight}{-10cm}   

\end{document}